\documentclass[reqno]{amsart}

 \usepackage[usenames,dvipsnames]{pstricks}
 \usepackage{epsfig}

\usepackage{epsfig,amstext,amsmath,amsfonts,amssymb,float}
\usepackage{graphicx}
\usepackage[nospace]{cite}
\usepackage{fancybox,array}
\usepackage{color}

\newcommand{\sm}{\setminus}
\newcommand{\real}{\mathbb{R}}

\newcommand{\stwo}{\mathbb{S}^2}

\newcommand{\nat}{\mathbb{N}}

\newcommand{\ep}{\epsilon}
\newcommand{\ffi}{\varphi}

\newcommand{\intstwo}{\int_{\mathbb{S}^2}}

\newcommand{\wto}{\rightharpoonup}

\newcommand{\la}{\langle}
\newcommand{\ra}{\rangle}

\newcommand{\ow}{w}
\newcommand{\vt}{\widetilde{U}}
\newcommand{\dif}{\mathrm{d}}
\newcommand{\diff}{\,\mathrm{d}}
\newcommand{\e}{\mathrm{e}}

\newcommand{\norm}[1]{\ensuremath{\left\Vert #1 \right\Vert }}

\DeclareMathOperator \vect{span}
\DeclareMathOperator \rge{rge}

\let\le=\leqslant
\let\ge= \geqslant

\newcommand{\abs}[1]{\left\vert #1 \right\vert}

\newcommand{\caD}{{\mathcal D}}
\newcommand{\caE}{{\mathcal E}}
\newcommand{\caF}{{\mathcal F}}

\newcommand{\caM}{{\mathcal M}}

\newcommand{\caO}{{\mathcal O}}

\newcommand{\caS}{{\mathcal S}}

\newcommand{\caU}{{\mathcal U}}

\newcommand{\bbN}{{\mathbb N}}

\newcommand{\bbR}{{\mathbb R}}

\def\XXint#1#2#3{{\setbox0=\hbox{$#1{#2#3}{\int}$ }
\vcenter{\hbox{$#2#3$ }}\kern-.59\wd0}}

\newtheorem{thm}{Theorem}
\newtheorem{A}{Assumption}
\newtheorem{lem}[thm]{Lemma}
\newtheorem{prop}[thm]{Proposition}

\newtheorem{rem}{Remark}

\newenvironment{proofof}
{\medskip\noindent{\it Proof of}}{\nolinebreak\hfill$\Box$\medskip}

\begin{document}

\title[Mean-field and phase transitions for nematic liquid crystals]
{Mean-field limit and phase transitions for nematic liquid crystals in the continuum}

\author[S. Bachmann]{Sven Bachmann}
\address{Mathematisches Institut der Universit\"at M\"unchen, Theresienstr.~39, D-80333
M\"unchen, Germany}
\email{sven.bachmann@math.lmu.de}

\author[F. Genoud]{Fran\c cois Genoud}
\address{Delft Institute of Applied Mathematics \\
Delft University of Technology \\
Mekelweg 4 \\
2628 CD Delft \\ The Netherlands}
\email{S.F.Genoud@tudelft.nl}

\begin{abstract}
We discuss thermotropic nematic liquid crystals in the mean-field regime. 
In the first part of this article, we rigorously carry out the mean-field limit 
of a system of $N$ rod-like particles as $N\to\infty$,
which yields an effective `one-body' free energy functional. 
In the second part, we focus on spatially homogeneous systems, for which 
we study the associated Euler--Lagrange equation, 
with a focus on phase transitions for general axisymmetric potentials. 
We prove that the system is isotropic at high temperature, while anisotropic 
distributions appear through a transcritical bifurcation as the temperature is lowered. 
Finally, as the temperature goes to zero we also prove, in the concrete case of 
the Maier--Saupe potential, that the system converges to perfect nematic order.
\end{abstract}

\maketitle

\begin{center}
\emph{Dedicated to Charles-Edouard Pfister, a mentor and a friend.}
\end{center}

\section{Introduction}\label{Intro}

Besides their obvious and ubiquitous importance in practical applications, liquid crystals are of considerable theoretical and mathematical interest. 
There is a temperature range in which these large molecules behave on the one hand like fluids, in the sense that their 
centres of mass respect translation invariance at the macroscopic level, but on the other tend to align their orientations, thereby breaking the rotational invariance. In other words, there is no long-range positional order, while orientational order may occur. In this article we consider so-called {\em thermotropic} liquid crystals, where the phase transition is driven by temperature. We further restrict our study to the case of {\em nematic} liquid crystals,
which are formed of long rod-like molecules, and tend to have long-range orientational order
parallel to their long axis. We refer the reader to \cite{Gennes:1995aa} for more details on the physics of liquid crystals.

The goal of this article is to derive an effective 
theory for the equilibrium state of nematic liquid crystals in the continuum from exact microscopic statistical mechanics. Physical considerations were used long ago to obtain such theories~\cite{Onsager:1949aa, Maier:1958aa, Maier:1959aa, Maier:1960aa, oseen,frank,ericksen_1976,ericksen_1990} --- see~\cite{Gennes:1995aa} for a complete overview --- but a mathematical control of the approximations involved has been lacking for decades. While the general programme of rigorously deriving effective equations for the thermal equilibrium state in the form of a \emph{mean-field limit} was carried out in~\cite{Messer:1982aa} in the classical case, in~\cite{Fannes:1980hd} in the quantum case and in~\cite{Petz:1989to,Raggio:1989gc} in the C*-algebraic setting, the concrete problem of liquid crystals was somewhat overlooked. Recent developments related to cold Bose gases have revived the subject \cite{Frohlich:2009be,Knowles:2010hy,Lewin:2014aa,Lewin:2015aa} and focused much more on the Gross--Pitaevskii limit \cite{Lieb:2000aa,Lieb:2005aa}. Parallel results on the continuum limit from lattice models~\cite{Cicalese:2009aa} have led to different effective theories that we will not address here. 
A few exact results about liquid crystals modelled by classical dimers or $k$-mers on lattices have also been obtained, see e.g.~\cite{heilmann1972,Disertori:2013ff,Giuliani:2014aa}. In the continuum, the appearance of orientational order at low temperature and high density was proved in~\cite{Gruber:2002ta,Romano:1998uw}.

Effective theories for liquid crystals at equilibrium 
have been derived at essentially two different levels. At the
{\em statistical level} \cite{Onsager:1949aa, Maier:1958aa, Maier:1959aa, Maier:1960aa}, the central object of the theory is the distribution of orientations of a molecule immersed in an
averaged `molecular-field' potential due to all other molecules. Such theories, derived by formal arguments from microscopic models, typically describe
bulk systems, with no account taken of boundary conditions, external fields, etc. Phase 
transitions are studied via an effective free energy, expressed by means of
an averaged order parameter revealing the anisotropy of the system. 
On the other hand,
at the {\em macroscopic level},
phenomenological theories \`a la Landau describe liquid crystals in terms of
a vector or tensor field of order parameters characterising the long-range orientational
order of the system. The free energy is then a functional of such
fields, to be minimised by variational methods, and giving rise to nonlinear 
partial differential equations (PDE) describing the liquid crystal phase,
see \cite{oseen,frank,ericksen_1976,ericksen_1990} and references
therein. These models have proved successful to describe numerous physical scenarios,
including external fields, boundary conditions and domain walls, and are routinely used in applications.
The relevant PDEs/minimisation problems 
have been the object of a substantial amount of work, 
see e.g.~the review paper \cite{ll} or the more recent contribution \cite{mz}.
Finally, the connection between theories at the statistical level and 
the macroscopic level has also been investigated by several authors, see 
e.g.~\cite{Katriel:1986gr} for a statistical physics approach, and \cite{bm,han} for 
PDE/variational arguments. 
To summarise, phenomenological macroscopic theories
are reasonably well understood mathematically and the relation of the latter with statistical
theories has already seen significant progress. However, the foundations of the statistical theories themselves, despite being physically clear, is mathematically mostly heuristic.
We shall take here a first step towards bridging this gap.

Concretely, we consider a particular limiting regime of a system of $N$ rod-like molecules confined to a fixed bounded domain $\Lambda\subset\bbR^d$, as $N\to\infty$. The particles interact with a two-body potential whose range is of the order of the size of the container, and the balance between minimising the total energy $\caE_N$ of the system and maximising the entropy $\caS_N$ (multiplied by the inverse temperature) may lead to the nematic phase transition briefly described above. From a thermodynamics perspective, both energy and entropy are expected to be extensive quantities, namely $\caE_N = \caO(N)$ and $\caS_N = \caO(N)$. Since the number of pairs of particles grows quadratically with $N$, it is a natural Ansatz to scale the strength of the two-body interaction by $N^{-1}$ as the limit $N\to\infty$ is carried out. We will refer to this as the mean-field scaling: this is the regime of very many, very weak collisions as the number of particles diverges.

The exact form of the potential will not matter very much for most of the following, but the nature of the liquid crystals is reflected in the fact that the potential depends not only on the positions $x,x'\in\Lambda$ of the two particles, but also on their orientations $m,m'\in M$, where $M$ is (typically) the upper half-sphere. As is physically relevant, the angular part of the potential shall have a minimum which may induce an alignment of the molecules. For technical reasons, singularities can only be repulsive in our model, 
and in fact, the full potential is bounded below. Concretely, in three dimensions, we have in mind the example
\begin{equation*}
\frac{A_N\exp(-\xi_N\vert x - x' \vert)}{\xi_N \vert x - x' \vert}\sin^2(\measuredangle(m,m')),
\end{equation*}
for some $A_N,\xi_N>0$, and with
$\measuredangle(m,m')$ the angle between directions $m$ and $m'$. Here, the mean-field scaling corresponds to choosing $\xi_N = \xi$, while $A_N = N^{-1} A$, for some fixed $A,\xi>0$.

We will prove that the full equilibrium statistical mechanics of the system reduces to a `one-body' problem as $N\to\infty$.  
The Gibbs postulate states that the equilibrium states are completely characterised as minimisers of the free energy functional $\caF_N$, namely the difference between energy and entropy. Here we will show that, in the scaling limit, such minimisers converge to a superposition of product states that minimise a one-particle free energy functional.  
This `mean-field functional' is easily derived by formally computing the free energy density associated to a product state in the limit $N\to\infty$.

The second part of this article is devoted to the problem of the nematic phase transition itself in the framework of 
the effective theory. Here and for the rest of the article, we shall restrict our attention to spatially homogeneous solutions. For this purpose we study the solution set of the Euler--Lagrange equation associated with the effective free energy functional, as a function of the inverse temperature. In the spatially homogenous case, any minimiser --- namely any thermal equilibrium state --- is a solution of the nonlinear self-consistency equation
\begin{equation}\label{self}
\nu(m) = Z_\beta(\nu)^{-1} \exp\left(-\beta  \int_M U(m,m')\nu(m')\diff m'\right),\quad Z_\beta(\nu) := \int_M \e^{-\beta \int_M U(m,m')\nu(m')\diff m'} \diff m,
\end{equation}
where $U(m,m')$ is the angular part of the potential, and $\beta$ is the inverse temperature, $\beta = (kT)^{-1}$, where $k$ is Boltzmann's constant and $T$ the temperature. It is easy to observe, and physically clear, that at high temperature the system has a unique thermal state. Namely, there is a unique minimiser, given by the homogeneous and isotropic distribution. This is the disordered phase, corresponding to a branch of solutions existing for all values of the temperature --- referred to as the `isotropic branch' below ---, but yielding the stable state at high temperatures only.

We then carry out a local bifurcation analysis from the isotropic branch around a temperature $T_\star$ at which a transcritical bifurcation occurs, under general sufficient conditions. This yields the existence of additional branches of solutions around this point, consisting of non-isotropic states. These indicate that the molecules favour alignment along a particular direction, commonly referred to as the `director'. Since the orientation of the director is arbitrary, the result shows that, for low enough temperatures, the rotation symmetry is broken and there exists an infinite number of equilibrium states: this is the mathematical condition for a phase transition.

Finally, we consider a specific angular potential originally introduced in the discussion of liquid crystals by Maier and Saupe~\cite{Maier:1958aa}. In fact, since rods carry an orientation but no direction the first nontrivial term in a multipole expansion of any molecular potential is the quadrupole, of which the Maier--Saupe interaction is the axially symmetric part. For this very specific but physically relevant interaction, we compute explicitly the transition temperature $T_\star$ as a function of the parameters of the model, recovering the result of~\cite{slast}, where the transcritical bifurcation was already pointed out.
In the context of `Onsager interactions', a transcritical bifurcation was also exhibited in \cite{vollmer}.

In the last part of our analysis, we study the behaviour of the equilibrium states in
the limit of small temperatures. We observe that, in a precise sense, they converge to the pure,
completely aligned state characterising the system at zero temperature. We analyse the phase diagram in terms of the typical order parameter given by the statistical average $\langle\sin^2(\theta)\rangle$,
where $\theta$ is the angle between a molecule and the director. For this observable, we exhibit a differentiable branch converging to $0$ --- which characterises a perfect (prolate) nematic state --- as the temperature is lowered to $0$ 
(i.e.~as $\beta\to\infty$). This yields an asymptotic behaviour that can be inferred from the explicit solutions of~\cite{slast,const}. The proof we present is based on Laplace's method and the implicit function theorem, 
a strategy we believe could be useful to handle more general axisymmetric potentials than the 
Maier--Saupe interaction, and more general order parameters.

Both the transcritical bifurcation occurring at $T_\star$ and the zero temperature limit illustrate, in a rigorous mathematical formulation, the bifurcation diagram obtained in \cite{Maier:1958aa,Maier:1959aa}, 
which we discuss at a heuristic level
at the end of Section~\ref{Results}. However, most of our results extend beyond the Maier--Saupe potential. The analysis is therefore more involved, as we cannot reduce the discussion of the phase diagram to the behaviour of the eigenvalues of the quadrupole matrix, as was done in \cite{slast,const}.

Before going into more details, we would like to emphasise that the mean-field limit in classical statistical mechanics was essentially treated already by Messer and Spohn in~\cite{Messer:1982aa}, where the associated Euler--Lagrange equation is called the `isothermal Lane--Emden equation'. The proof we present here follows closely \cite{Messer:1982aa} --- even though we allow for repulsively unbounded interactions, which were not considered there --- and the even earlier results of~\cite{robinson1967} about the entropy density. Similar arguments were used in~\cite{Kiessling:1989vd} to discuss the thermodynamics of gravitating systems, and further in~\cite{Caglioti:1992aa,Lions:2000uh} in the different context of vortices of the Euler and Navier--Stokes equations. 
We are also indebted to~\cite{Rougerie:2014aa} for a clear overview of the subject.

The rest of the paper is organised as follows. In Section~\ref{Results}, we introduce the basic mathematical and statistical-mechanical setup to describe the model and state our results, namely the convergence in the mean-field limit, the abstract bifurcation analysis, and the concrete existence of a nematic phase transition in the case of the Maier--Saupe model. Section~\ref{Proofs} then presents the proofs. 

\subsection*{Acknowledgement}The authors would like to thank Margherita Disertori for very helpful discussions. 
They are also indebted to an anonymous referee
who pointed out a mistake in an earlier version of the manuscript, and helped improve 
the general presentation of the paper.


\section{Setting and results}\label{Results}

In this section we lay down the mathematical formalism we shall use, and we give a unified
presentation of our main results, which will be proved in Section~\ref{Proofs}.


\subsection{Mean-field limit}
We consider a system of $N$ classical identical particles confined in a connected, 
compact domain $\Lambda\subset\bbR^d$ and
carrying an internal degree of freedom. 
The configuration space $\Omega^N$ is given by
\begin{equation*}
\Omega^N := (\Lambda\times M)^N,
\end{equation*}
where $M$ is a smooth, connected, compact manifold. 
For the application we have in mind, namely nematic liquid crystals, $M = \bbR P^{d-1}$, the real projective space. We shall denote by $(x,m)$ points in $\Lambda\times M$, while a point in the full configuration 
space will be 
\[
X_1^N = (X_1,\ldots,X_N)\in \Omega^N, \quad X_i=(x_i,m_i), \quad i=1,\dots,N.
\] 
Similarly, for $1\le j \le k \le N$, we shall also use the shorthand notation
\begin{equation*}
X_j^k = (X_j,\dots,X_k), \qquad \dif X_j^k = (\dif X_j,\dots,\dif X_k).
\end{equation*}

An $N$-particle state is a probability measure over $\Omega^N$, the set of which we denote $\caE(\Omega^N)$. The finite volume, finite particle number equilibrium state is given by the Gibbs measure at 
inverse temperature $\beta = (kT)^{-1}$, where $k$ is Boltzmann's constant and $T$ is the temperature. 
Explicitly, this is the following absolutely continuous probability measure $\mu_{N,\beta}$
on $\Omega^N$ with density $\rho_{N,\beta}$:
\begin{equation*}
\mu_{N,\beta}( \dif X_1^N) = \rho_{N,\beta}(X_1^N)\diff X_1^N = \frac{1}{Z_{N}(\beta)} \exp\left(-\beta V_{N}(X_1^N)\right)\dif X_1^N,
\end{equation*}
where 
\[
Z_{N}(\beta) = \int_{\Omega^N}\exp\left(-\beta V_{N}(X_1^N)\right) \diff X_1^N
\]
is the partition function, ensuring normalisation of the equilibrium state.
The interaction energy of the particles in a configuration $X_1^N$ is given by the two-body potential
\begin{equation*}
V_{N}(X_1^N) := \frac{1}{N-1}\sum_{1\le i<j\le N} v\left(\abs{x_i-x_j}; m_i,m_j\right),
\end{equation*}
which we shall refer to as the \emph{mean-field scaling}. We will carry out the mean-field limit, $N\to\infty$, 
under the following assumption, where $D=\sup_{x,y\in\Lambda}|x-y|$ is the diameter of the region $\Lambda$.
\begin{A}\label{A:VBasic}
The function $v:(0,\infty)\times M\times M \to\bbR$
satisfies the following conditions:
\begin{enumerate}
\item {\rm Symmetry:} $v(r;m,m') = v(r;m',m)$.
\item {\rm Lower semi-continuity:} for all $(m,m')\in M\times M$, there holds 
$$r=\lim_{n\to\infty} r_n \implies v(r;m',m) \le \liminf_{n\to\infty} v(r_n;m',m).$$
\item {\rm Integrability:} for all $(m,m')\in M\times M$,
$r\mapsto r^{d-1}v(r,m,m')$ is integrable over $(0,D)$, and the function
$J(m,m'):=\omega_d \int_0^D r^{d-1} |v(r,m,m')| \diff r$ is bounded on $M\times M$, 
where $\omega_d$ is the surface of the unit sphere $\mathbb{S}^{d-1}$.
\end{enumerate}
\end{A}
It follows from Assumption~\ref{A:VBasic} that $V_{N}$ is invariant under permutation, i.e.
$$V_N(X_1,\ldots X_N) = V_N(X_{\pi(1)},\ldots X_{\pi(N)})$$ for any $\pi$ in the symmetric group, 
and lower semi-continuous with respect to $x_1^N$. In the following, 
we shall sometimes write $v(X,Y)$ for $v(\vert x-y\vert;m,m')$, which is a 
lower semi-continuous function of $(x,y)\in\Lambda\times\Lambda$.

\begin{rem}
\rm
Since $\Omega$ is compact, lower semi-continuity implies boundedness below. 
In particular, singularities can only be repulsive. Furthermore,
\begin{equation*}
\frac{1}{Z_{N}(\beta)} \exp\left(-\beta V_{N}(X_1^N)\right) = \frac{1}{\tilde Z_{N}(\beta)} \exp\big(-\beta \tilde V_{N}(X_1^N)\big),
\end{equation*}
where $\tilde Z_{N}, \tilde V_N$ are obtained from the replacement $v\mapsto \tilde v = v-\min v$. 
Therefore, in order to simplify the proof of Theorem~\ref{effective theory}, {\em we will assume 
without loss of generality that $v\ge 0$.}
\end{rem}

Finally, we define the reduced $k$-particle density by
\begin{equation*}
\rho_{N,\beta}^k(X_1^k) := \int_{\Omega^{N-k}}\rho_{N,\beta}(X_1^N)\diff X_{k+1}^N, 
\quad 1\le k\le N,
\end{equation*}
namely the marginal density of $\mu_{N,\beta}$, still normalised to be a probability density. 
Note that the choice of indices to be integrated out is irrelevant,
thanks to the symmetry assumption on the potential.

For any finite $N\in\bbN$, the Gibbs measure $\mu_{N,\beta}$ is the unique minimiser among
all probability measures $\nu_N$ on $\Omega^N$ of the free energy functional
\begin{equation*}
\caF_{N,\beta}(\nu_N):= \caE_N(\nu_N) - \beta^{-1} \caS_{N}(\nu_N),
\end{equation*}
where
\begin{align*}
\caE_N(\nu_N) &:= \int_{\Omega^N} V_{N}(X_1^N) \diff\nu_N(X_1^N), \\
\caS_{N}(\nu_N) &:= \begin{cases} - \int_{\Omega^N} \nu_N(X_1^N) \ln\nu_N(X_1^N) \diff X_1^N 
&\text{if }\nu_N\text{ is absolutely continuous}, \\ -\infty &\text{otherwise}. \end{cases}
\end{align*}
Namely, $\caF_{N,\beta}(\nu_N)$ is the difference between the \emph{total energy} and 
the \emph{entropy} of the system at inverse temperature $\beta$. We let 
\[
F_{N,\beta} :=\caF_{N,\beta}(\mu_{N,\beta})
\] 
so that $F_{N,\beta}\le \caF_{N,\beta}(\nu_N)$ for all probability measures $\nu_N$ by the variational principle.

Note that while the entropy depends on the full measure, 
the energy functional depends only on its second marginal, namely
\begin{equation}\label{energy}
\int_{\Omega^N} V_{N}(X_1^N) \diff \nu_N(X_1^N)
= \frac{N}{2}\int_{\Omega^2} v\left(\abs{x-y}; m,m'\right) \diff\nu_N^2(X,Y),
\end{equation}
where we used the symmetry of the two-body potential $v$. 
For a product measure $\nu_N = \nu^{\otimes N}$, the free energy density reads
\begin{equation}\label{product}
N^{-1}\caF_{N,\beta}(\nu^{\otimes N}) = \frac{1}{2}\int_{\Omega^2} v(X,Y) \diff\nu(X)\diff\nu(Y)  
- \beta^{-1} \caS_1(\nu)=:\mathfrak{f}_\beta(\nu).
\end{equation}

We are now equipped to state the results of this paper. 
The first theorem states that minimising $\caF_{N,\beta}$ over product measures, 
which a priori only yields an upper bound on the global minimum for all $N$, 
in fact produces the exact thermal states as $N\to\infty$. 
\begin{thm}\label{effective theory}
Fix $0<\beta<\infty$.
With the notation above, let $\caM_\beta$ be the set of minimisers of $\mathfrak{f}_\beta$, 
and let $f_\beta$ be the corresponding minimum. If Assumption~\ref{A:VBasic} holds, then
\begin{equation*}
\lim_{N\to\infty} N^{-1}F_{N,\beta} = f_\beta.
\end{equation*}
Furthermore, if $\mu_{\infty,\beta}$ is an accumulation point of $\mu_{N,\beta}$ in the weak-* topology, then
there is a probability measure $P_\beta$ supported on $\caM_\beta$ such that
\begin{equation*}
\mu^{k}_{\infty,\beta} = \int_{\caM_\beta}\nu^{\otimes k}\diff P_\beta(\nu), \quad \text{for all} \ k\in\bbN.
\end{equation*}
Finally, $\mu^{k}_{\infty,\beta}$ is an absolutely continuous measure on $\Omega^k$ for all $k\in\bbN$. 
In particular, $P_\beta$ is supported on absolutely continuous measures.
\end{thm}

Let us recall here that 
an accumulation point of $\{\mu_{N,\beta}\}_{N\in\bbN}$ in the weak-* topology is a probability measure $\mu$
on $\Omega$ such that there exists a subsequence of indices $\{N_j\}_{j\in\nat}$ for which
$\int_\Omega f \diff\mu_{N_j} \to \int_\Omega f \diff\mu$ as $j\to\infty$, for all continuous
functions $f:\Omega\to\real$. We shall write $\mu_{N_j}\wto\mu$ for this notion of convergence,
also known as `weak convergence of measures'.

Note that the set of minimisers is not empty: minus the entropy is a 
lower semi-continuous function on the set of probability measures 
over $\Omega$ equipped with the weak-* topology. 
The same holds for the energy if the potential is lower semi-continuous, 
see~(\ref{e:lsc}) below. Therefore, the free energy is a lower 
semi-continuous function defined on a compact set, 
namely the set of probability measures over $\Omega$, and hence it reaches its infimum.

In the next proposition, we exhibit the Euler--Lagrange equation associated with the functional $\mathfrak{f}_\beta$, 
which is solved in particular by its minimisers.
The counterpart of this equation in~\cite{Messer:1982aa} is referred to as the Lane--Emden equation.
\begin{prop}\label{prop:Gap Eq}
The absolutely continuous critical points of the mean-field functional $\mathfrak{f}_\beta$  
have a density $\nu(x,m)$ satisfying
\begin{equation}\label{GapEq}
\nu(x,m) = \frac{\e^{-\beta H_\nu(x,m)}}{\int_\Omega \e^{-\beta H_\nu(x',m')}\diff x'\diff m'}, 
\quad\text{a.e.} \ (x,m)\in\Omega,
\end{equation}
where
\begin{equation*}\label{H_nu}
H_\nu(x,m) := \int_\Omega v(x,x';m,m') \nu(x',m') \diff x'\diff m'.
\end{equation*}
\end{prop}

In the rest of the paper, we study the structure of the set of solutions to \eqref{GapEq}, in the
case when the system is {\em spatially homogeneous}. That is, instead of \eqref{GapEq}, 
we shall henceforth consider the equation
\begin{equation}\label{gap_hom}
\nu(m) = \frac{\e^{-\beta  H_{\nu}(m)}}
{\int_M \e^{-\beta H_{\nu}(m)} \diff m},
\qquad \text{with} \quad  H_{\nu}(m) = \int_M U(m,m')\nu(m')\diff m',
\end{equation}
where we consider a general two-body angular interaction $U\in L^\infty(M\times M)$.
We will refer to \eqref{gap_hom} as the \emph{self-consistency equation} (SCE),
following a common terminology in physics, for instance in the liquid crystal literature.
Noteworthily, steady-state solutions of the Smoluchowski equation, 
which arises in the kinetic theory of liquid crystals, are also solutions of the SCE, 
see \cite{const}.
\begin{rem}
\rm
Here, we implement the homogeneity assumption as a restriction on the set of solutions, yielding a simplified SCE. This is justified from a physical point of view in the limit of very long rods as was observed for example in~\cite{Onsager:1949aa,Lee:1986vi}. Also, in the very particular case of a two-body potential of the form
\begin{equation*}
v(x,x';m,m') = w(x,x') + U(m,m'),
\end{equation*}
the numerator in the right-hand side of \eqref{GapEq} factorises as
\begin{equation*}
\exp(-\beta H_\nu(x,m)) =
\exp\left(-\beta\int_\Lambda w(x,x') \nu_M(x')\diff x'\right)
\exp\left(-\beta\int_M U(m,m')\nu_\Lambda(m')\diff m'\right),
\end{equation*}
where $\nu_M, \nu_\Lambda$ are the marginals of $\nu$, i.e.
\begin{equation*}
\nu_M(x)=\int_M \nu(x,m) \diff m, \qquad \nu_\Lambda(m)=\int_\Lambda \nu(x,m) \diff x.
\end{equation*}
Hence, integrating \eqref{GapEq} over $\Lambda$ and dropping the index $\Lambda$ of $\nu_\Lambda$,
one gets exactly \eqref{gap_hom}.
\end{rem}


\subsection{Phase transitions}

We will now discuss (for spatially homogeneous systems) the phase diagram 
in the mean-field theory, by studying solutions of the SCE \eqref{gap_hom}. 
We shall call {\em solution} of \eqref{gap_hom} a pair $(\beta,\nu)\in(0,\infty)\times L^1(M)$ satisfying 
\eqref{gap_hom}. Observe that $\nu$ is then {\em automatically} positive and normalised,
hence the density of a probability measure on $M$. Note, furthermore, that any solution 
$(\beta,\nu)\in(0,\infty)\times L^1(M)$ of \eqref{gap_hom} necessarily has $\nu\in L^2(M)$. 
Therefore, we do not loose generality by considering solutions in 
$(0,\infty)\times L^2(M)$, which presents several technical advantages.

An important preliminary remark is that, under Assumption~\ref{const_mean.as} below, 
the isotropic state $(\beta,\nu\equiv|M|^{-1})$ is a solution of \eqref{gap_hom} for all $\beta>0$.
We first show that, for $\beta$ sufficiently small (high temperature),
this is in fact the only solution of \eqref{gap_hom}. 
We then exhibit a transition temperature $\beta_\star$, where the isotropic state undergoes
a transcritical bifurcation.
The last part of our analysis will focus on the special case of Maier--Saupe interactions.
In this context we can compute explicitly the transition temperature $\beta_\star$. 

We finally study the zero-temperature limit by means of the standard order parameter $\langle\sin^2(\theta)\rangle$,
where $\theta$ is the angle between a molecule and the director, and $\langle\cdot\rangle$
denotes statistical average. For this observable, we exhibit a differentiable branch converging to zero 
as $\beta\to\infty$, showing that the system freezes in a perfectly aligned state at zero temperature.

\medskip
Our first general results about the solutions of \eqref{gap_hom} will be formulated under
\begin{A}\label{const_mean.as}
\rm
$\int_M U(m,m')\diff m'$ is independent of $m\in M$.
\end{A}
Besides being mathematically necessary, this assumption is physically clear: 
the effective potential $H_{\vert M\vert^{-1}}$ is constant for the uniform distribution. 
In other words, a completely disordered system does not generate any force towards alignment.


\subsubsection{High temperature and transcritical bifurcation}
 
We will show the existence of a transition temperature $\beta_\star$, and of a branch of 
non-isotropic solutions of \eqref{gap_hom} crossing the {\em isotropic line} 
$$
\{(\beta,|M|^{-1}): \beta>0\}\subset (0,\infty)\times L^2(M)
$$
at the point $(\beta_\star,|M|^{-1})$.
To make the bifurcation analysis more transparent,
we shall conveniently reformulate \eqref{gap_hom} as an operator
equation by introducing the mapping 
$\Phi:(0,\infty)\times L^2(M) \to L^2(M)$ defined as
\begin{equation}\label{Phi}
\Phi(\beta,\nu)=\nu-\frac{\e^{-\beta H_\nu}}{\int_{M} \e^{-\beta H_\nu}},
\end{equation}
so that \eqref{gap_hom} becomes
\begin{equation}\label{operator_gap}
 \Phi(\beta,\nu)=0.
\end{equation}

The linearisation of \eqref{operator_gap} will play a central role in the analysis.
We will establish in Section~\ref{Proofs} that
\begin{equation*}
D_\nu \Phi(\beta,|{M}|^{-1})\mu = \mu-\beta K\mu,
\end{equation*}
where $D_\nu \Phi$ denotes the Fr\'echet derivative of $\Phi$ with respect to $\nu$, and 
$K:L^2({M})\to L^2({M})$ is the compact self-adjoint operator defined by
\begin{equation}\label{K}
(K\mu)(m):=-|{M}|^{-1}\int_M \vt(m,m')\mu(m')\diff m', \quad \mu\in L^2({M}).
\end{equation}
Here, recalling that $\int_M U(m,m'')\diff m''$ is constant by Assumption~\ref{const_mean.as},
we have defined 
\begin{equation}\label{dev}
\vt(m,m'):=U(m,m')- |M|^{-1}\int_M U(m,m'')\diff m''.
\end{equation} 
The linearisation of \eqref{operator_gap} is thus given by
\begin{equation}\label{linearisation}
\mu=\beta K\mu.
\end{equation}

The following proposition establishes the absence of phase transitions at high temperature, under the running assumption of spatial homogeneity.

\begin{prop}\label{high_temp.prop}
If Assumption~\ref{const_mean.as} holds and $\beta<\norm{K}^{-1}$, then there is no bifurcation points for \eqref{operator_gap}. 
\end{prop}

\begin{rem}\label{uniqueness.rem}
\rm
Note that Proposition~\ref{high_temp.prop} does not preclude the existence of 
non-isotropic states `far away' from the isotropic line; it only ensures that such solutions cannot
emerge from, or cross the isotropic line at points $(\beta,|M|^{-1})$ with $\beta<\norm{K}^{-1}$.
However, following the proof
of \cite[Theorem~3]{Messer:1982aa}, one can show that \eqref{operator_gap} indeed has
a unique solution (the isotropic state) for all $\beta<(2\|U\|_{L^\infty})^{-1}$. 
In the case of Maier--Saupe interactions 
discussed below, this is consistent with the picture obtained in Section~3.1 of \cite{slast}.
\end{rem}

We now formulate a general result yielding a local continuous curve of solutions of  
\eqref{operator_gap} bifurcating from an isotropic solution $(\beta_\star,|{M}|^{-1})$,
where the transition temperature $\beta_\star>0$ is characterised in terms
of the linear problem \eqref{linearisation}.

\begin{thm}\label{transition.thm}
Suppose that Assumption~\ref{const_mean.as} holds and that there exists $\beta=\beta_\star>0$ for which
\eqref{linearisation} has a one-dimensional solution space, spanned by some $\mu_\star\in L^2(M)\sm\{0\}$.
Then there exist $\ep>0$ and a continuous map 
$$
(-\ep,\ep)\ni s\mapsto (\beta(s),\nu(s))\in(0,\infty)\times L^2(M)
$$
such that
$(\beta(s),\nu(s))$ is a non-isotropic solution of \eqref{operator_gap} for $s\neq0$, with
$$
\beta(0)=\beta_\star \qquad\text{and}\qquad \nu(0)=|M|^{-1}.
$$
Furthermore, $\nu(s)=|M|^{-1}+s(\mu_\star+\rho(s))$, where $\rho$ 
is a continuous function of $s\in(-\ep,\ep)$ such that:
$$
\rho(0)=0 \qquad\text{and}\qquad \rho(s)\in[\mu_\star]^\perp, \quad\text{for all} \ s\in(-\ep,\ep).
$$
\end{thm}

Here, $[\mu_\star]^\perp$ denotes the orthogonal complement of $\mu_\star$ in $L^2(M)$.
Since $\int_M\vt(m,m')\diff m'\equiv0$, a non-zero solution $\mu_\star$ of \eqref{linearisation}
cannot be constant. Given the form of $\nu(s)$, this indeed 
ensures that each $\nu(s)$ is non-isotropic for $s\neq0$.

Obviously, finding the transition temperature $\beta_\star$ in Theorem~\ref{transition.thm}
(if it exists) is very much dependent
upon the specific choice of potential appearing in \eqref{K}. We conclude this section
with an explicit result in the case of {\em Maier--Saupe interactions}, i.e.~with the potential
\begin{equation}\label{MSpotential}
U(m,m')=w[1-P_2(\cos\gamma)],
\end{equation}
where $w>0$ is a coupling constant, and
$P_2$ is the second Legendre polynomial, $P_2(x)=(3x^2-1)/2$, and $\gamma = \measuredangle(m,m')$.
Now, a straightforward calculation shows that $U(m,m')$ satisfies
\begin{equation}\label{constancy}
\int_M U(m,m')\diff m'=2\pi\ow, \quad \text{for all} \ m\in M,
\end{equation}
so that Assumption~\ref{const_mean.as} is satisfied. Therefore, the existence of a
bifurcating branch follows from Theorem~\ref{transition.thm}, provided one can solve
the linearised equation \eqref{linearisation}.

\begin{prop}\label{kernel.prop}
Consider the Maier--Saupe potential~(\ref{MSpotential}). 
There exists a unique $\beta=\beta_\star(\ow)>0$ for which
\eqref{linearisation} has a one-dimensional solution space.
The value of $\beta_\star$ and a corresponding eigenvector $\mu_\star$ are explicitly given by
\begin{equation}\label{eigen}
\beta_\star=5/\ow, \qquad
\mu_\star(\theta,\varphi)= \mu_\star(\theta)=3 \cos^2\theta-1.
\end{equation}
Furthermore, the bifurcation at the point $(\beta_\star,|M|^{-1})$ given by  
Theorem~\ref{transition.thm} is transcritical,
in the sense that there are solutions $(\beta(s),\nu(s))$ bifurcating from 
$(\beta_\star,|M|^{-1})$ both with $\beta(s)<\beta_\star$ and with $\beta(s)>\beta_\star$.
\end{prop}


\subsubsection{Low temperature}

Finally, we study the low temperature limit, $\beta\to\infty$, in the context of 
the Maier--Saupe interaction \eqref{MSpotential}, where we now fix $\ow=1$ for simplicity.
As we will see, in this case the effective potential reads
\begin{equation}\label{effpot}
H_\nu(\theta) =1 - \left\langle  P_2(\cos(\cdot))\right\rangle_{\nu}P_2(\cos\theta).
\end{equation}

Now consider the function $g(\theta)=\sin^2(\theta)$ and let
\begin{equation*}
\xi:= \nu(g) = \frac{2}{3}\left(1-\left\langle  P_2(\cos(\cdot))\right\rangle_{\nu}\right).
\end{equation*}
Up to a non-essential constant, this is the standard `order parameter' considered in the physical literature \cite{Maier:1958aa}. Clearly, $\xi\in[0,1]$, where the uniform distribution $\nu_0 = (2\pi)^{-1}$ yields $\xi = \frac{2}{3}$, while the completely aligned --- so-called `prolate' --- state, $\nu = \delta$ implies $\xi = 0$. The values $(\frac{2}{3},1]$ represent states where the rods tend to be orthogonal to a particular direction --- so-called `oblate' states. 
We shall only consider the prolate state in the low temperature limit here.

With the above choice of order parameter $g$, the homogeneous SCE \eqref{gap_hom} reduces to 
\begin{equation}\label{IFTequ}
\xi = F(\beta,\xi):= \frac{\int_0^{\pi/2} \sin^3(\theta)\e^{-\beta (1-(3/2)\xi) (1-P_2(\cos\theta))}\diff \theta}{\int_0^{\pi/2} \sin(\theta)\e^{-\beta (1-(3/2)\xi) (1-P_2(\cos\theta))}\diff \theta}.
\end{equation}
Using a Laplace method argument, we will show that, as $\beta\to\infty$, the right-hand side of the equation converges to $\sin^2(0)$. This corresponds to the physical intuition: as the temperature decreases, the molecules align into the perfect nematic state. In fact, the following theorem shows that this picture remains essentially true at low enough temperatures.
\begin{thm}\label{zeroK.thm}
Let $G:(0,\infty]\times[0,\frac{2}{3})\to\bbR$ be
\begin{equation*}
G(\beta,\xi) :=\begin{cases}
\xi - F(\beta,\xi) & \text{if }\beta<\infty, \\
\xi & \text{if }\beta = \infty. \end{cases}
\end{equation*}
There exists $0<\bar\beta<\infty$, $\bar\delta>0$, and a continuously differentiable
function $\bar\xi: (\bar\beta,\infty)\to[0,\frac{1}{3})$ such that $\lim_{\beta\to\infty}\bar\xi(\beta) = 0$ and
\begin{equation*}
\left\{(\beta,\xi)\in (\bar\beta,\infty]\times[0,\bar\delta) : G(\beta,\xi) = 0\right\} = 
\left\{(\beta,\bar\xi(\beta)) : \beta\in(\bar\beta,\infty] \right\}.
\end{equation*}
Moreover, for any $\bar\xi\in[0,\bar\delta)$, there exists $\nu\in F$ such that $\xi = \nu(\sin^2(\cdot))$.
\end{thm}

Finally, we comment on the phase diagram based on Figure~\ref{fig:curves}, which mirrors a similar figure in~\cite{Maier:1959aa}, where we plot $\xi\mapsto F(\beta,\xi)$ for various values of $\beta$. The solutions of the SCE~(\ref{IFTequ}) for the order parameter $\xi$ are given by the intersection with the diagonal. As already discussed, the point $\xi_0 = 2/3 = \int_0^{\pi/2}\sin^3\theta\diff\theta$, which corresponds to the isotropic distribution, is a solution for all $0 < \beta < \infty$. At low $\beta$, it is in fact the unique solution, see Remark~\ref{uniqueness.rem}. As $\beta$ increases, two new branches of solutions $\xi_1(\beta),\xi_2(\beta)$ appear in a saddle-node bifurcation, with the coexistence of three solutions. Increasing $\beta$ further, the two anisotropic solutions drift away from each other until $\xi_2(\beta)$ crosses $\xi_0$ in the transcritical bifurcation exhibited in Proposition~\ref{kernel.prop}. From there on, $\xi_2>2/3$ becomes an unphysical solution, while $\xi_1$ converges, as $\beta\to\infty$, to its ideal nematic value, $\overline{\xi}(\infty) = 0 = \int_0^{\pi/2}\sin^3\theta\delta_0(\theta) \diff\theta$, in the notation of Theorem~\ref{zeroK.thm}. The nature of the critical points at the level of the measure cannot be concluded from the picture, but it can if~(\ref{IFTequ}) is interpreted as the derivative of an effective $(\beta,\xi)$-dependent free energy function, \`a la Landau. Indeed, the second derivative at the critical point is positive, namely $\xi$ corresponds to a physical equilibrium point, if $\partial_\xi F< 1$.
That is, $\xi$ is a stable solution of \eqref{IFTequ}. Therefore, $\xi_1$ is stable in its range of existence 
while $\xi_2$ is not. As for $\xi_0$ it is stable for $\beta<\beta_\star$ and becomes unstable as $\beta$ crosses $\beta_\star$. This exchange of stability is characteristic of a transcritical bifurcation.
\begin{center}
\begin{figure}[htb]
\includegraphics{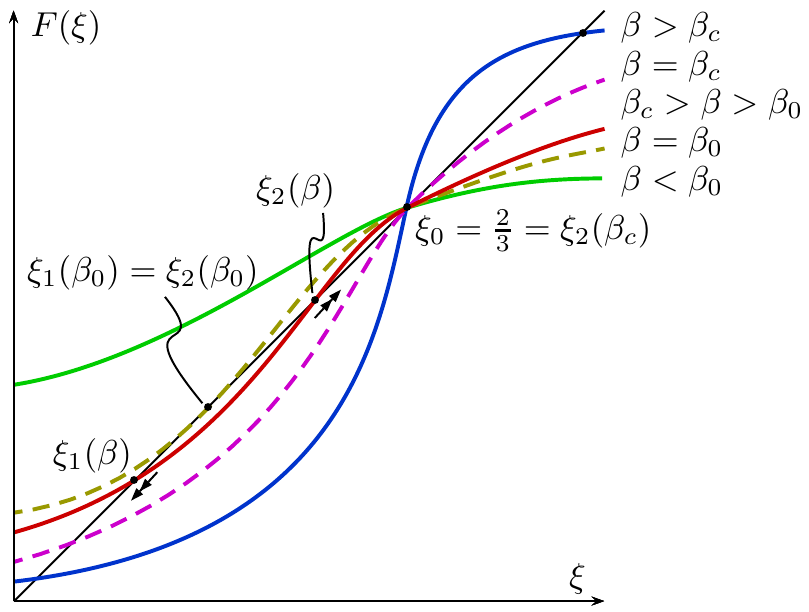}
\caption{The plot of $\xi\mapsto F(\beta,\xi)$ from~(\ref{IFTequ}) for different values of $\beta$. The curves are slightly exaggerated to improve readability.}
\label{fig:curves}
\end{figure}
\end{center}

\begin{rem}
\rm
Interestingly, the SCE in the case of the Maier--Saupe potential yields a pure differential equation for the effective potential $H_\nu$. Indeed, we multiply~(\ref{gap_hom}) by $U(m,m')$, integrate over $M$ and apply the Laplacian to get, still quite generally,
\begin{equation*}
\Delta H_\nu(m) = \frac{\int_M \Delta_m U(m,m')\e^{-\beta H_\nu(m')}\diff m'}{\int_M \e^{-\beta H_\nu(m')\diff m'}}.
\end{equation*}
If, now, $U$ is the Legendre polynomial $P_l$, then $- \Delta_m U(\cos\gamma) = l(l+1) U(\cos\gamma)$, so that
\begin{equation*}
- \Delta H_\nu = l(l+1) H_\nu.
\end{equation*}
Hence $H_\nu$ is proportional to $P_l$, as we shall see by an explicit calculation below. We also note that in the analogous case of the Coulomb potential, for which $\Delta_x V(x,y) = 4\pi \delta(x-y)$, the corresponding equation for the effective potential $\Psi(x)$ reads
\begin{equation*}
\Delta \Psi = 4\pi\kappa_\beta \exp(-\beta \Psi),
\end{equation*}
which is the well-known equation for the self-interacting gravitational potential of a mass distribution.
\end{rem}


\section{Proofs}\label{Proofs}

In this section we give complete proofs of the results stated above. 

\subsection{Mean-field limit}
First of all, a standard compactness argument yields the existence of limiting measures of Gibbs states 
as $N\to\infty$, and their symmetry implies a simple decomposition over product measures.
Indeed, since $\Omega = \Lambda\times M$ is a compact metric space, 
Tychonoff's theorem ensures that $\Omega^\infty := \times_{\bbN}\Omega$ is compact in the product topology.  
The measures $\mu_{N,\beta}$, defined originally on $\Omega^N$, 
can be extended by Hahn--Banach to $\Omega^\infty$. 
We shall denote the extension by $\mu_{N,\beta}$ again. 
Since the set of probability measures on a compact space is weakly-* compact, 
there are weak-* accumulation points, denoted $\mu_{\infty,\beta}$, that are again probability measures 
on $\Omega^\infty$. We write $\mu_{N_k}\rightharpoonup\mu_{\infty,\beta}$, as $k\to\infty$.

Since $\mu_{N,\beta}$ is symmetric for all $N\in\bbN$, namely 
$\mu_{N,\beta}(\dif X_1,\ldots, \dif X_N) = \mu_{N,\beta}(\dif X_{\pi(1)},\ldots, \dif X_{\pi(N)})$ 
for all permutations $\pi$ in the symmetric group with $N$ elements, so are the limiting measures. 
Hence, for any $\mu_{\infty,\beta}$, a theorem of Hewitt and Savage \cite{Hewitt:1955th} yields 
the existence of a unique probability measure $P$ on the set of states over the one-particle 
configuration space $\Omega$,  $\caE(\Omega)$, such that
\begin{equation*}
\mu_{\infty,\beta}^k = \int_{\nu\in\caE(\Omega)}\nu^{\otimes k}\diff P(\nu)
\end{equation*}
for any $k\in\bbN$. This decomposition of the weak accumulation points will play a crucial role in the proof of Theorem~\ref{effective theory}, once we have shown that the
limiting one-particle functional is affine.

Before we can give the proof of Theorem~\ref{effective theory}, 
we need the following properties of the entropy; see e.g.~\cite{Ruelle:1999aa} for proofs. Since they hold for arbitrary probability measures, we drop the index $N$, when not needed.
\begin{prop}
The entropy functional satisfies:
\begin{itemize}
\item[(a)] {\rm Negativity:} $\caS(\nu)\le 0$.
\item[(b)] {\rm Concavity:} $\caS(\lambda \nu + (1-\lambda) \nu') \ge \lambda\caS(\nu) + (1-\lambda)\caS(\nu')$.
\item[(c)] {\rm Upper semi-continuity:} $\limsup_{\nu\rightharpoonup\nu_0} \caS(\nu)\le \caS(\nu_0)$.
\item[(d)] {\rm Subadditivity:} for a Gibbs measure $\mu_{N,\beta}$,
\begin{equation}\label{subadditivity}
\caS_N(\mu_{N,\beta}) \le \left\lfloor\frac{N}{k}\right\rfloor \caS_k(\mu_{N,\beta}^k) + \caS_{N-k \lfloor\frac{N}{k}\rfloor }\big(\mu^{N-k \lfloor\frac{N}{k}\rfloor}_{N,\beta}\big).
\end{equation}
\end{itemize}
\end{prop}


\begin{proofof} {\em Theorem~\ref{effective theory}.}
Our proof essentially follows \cite{Messer:1982aa}, although we allow for unbounded interactions here.
For an arbitrary measure $\nu$ on $\Omega$, the product measures $\nu^{\otimes N}$ 
yield an $N$-independent free energy density (see~(\ref{product})), so that 
\begin{equation}\label{MF: upper}
N^{-1}F_{N,\beta} \le \mathfrak{f}_\beta(\nu).
\end{equation}
In particular $N^{-1}F_{N,\beta} \le f_\beta$.

It remains to prove an asymptotic lower bound. In the following, we consider a fixed converging subsequence $\mu_{N_j,\beta}\rightharpoonup\mu_{\infty,\beta}$. 
Recalling~(\ref{energy}), 
\begin{equation*}
\lim_{j\to\infty} \frac{\caE_{N_j}(\mu_{{N_j},\beta})}{{N_j}} = \lim_{j\to\infty}\frac{1}{2}\int_{\Omega\times\Omega} v(X,Y) \diff\mu_{N_j,\beta}^2(X,Y) = \frac{1}{2}\int_{\Omega\times\Omega} v(X,Y) \diff\mu_{\infty,\beta}^2(X,Y)
\end{equation*}
holds by definition of the convergence of measures if $v$ is continuous. 
We further recall (see e.g.~\cite{Billingsley:1968aa})
that the convergence still holds whenever $v$ is bounded and lower semi-continuous. Finally, if $v$ is unbounded, we consider the cut-off potential $v^\lambda(X):=\min\{v(X),\lambda\}$, for which $\liminf_{\lambda\to\infty}v^\lambda(X) = v(X)$. As $v(X)\ge v^\lambda(X)$,
\begin{equation*}
\lim_{j\to\infty}\int_{\Omega\times\Omega} v(X,Y) \diff\mu_{N_j,\beta}^2(X,Y) \ge \lim_{j\to\infty}\int_{\Omega\times\Omega} v^\lambda (X,Y) \diff\mu_{N_j,\beta}^2(X,Y) = \int_{\Omega\times\Omega} v^\lambda (X,Y) \diff\mu_{\infty,\beta}^2(X,Y).
\end{equation*}
Since this holds for any $\lambda>0$, it also does for the $\liminf_{\lambda\to\infty}$, so that
\begin{equation}\label{e:lsc}
\lim_{j\to\infty}\frac{\caE_{N_j}(\mu_{{N_j},\beta})}{{N_j}} \ge \liminf_{\lambda\to\infty}\frac{1}{2}\int_{\Omega\times\Omega} v^\lambda (X,Y) \diff\mu_{\infty,\beta}^2(X,Y) \ge \frac{1}{2}\int_{\Omega\times\Omega} v (X,Y) \diff\mu_{\infty,\beta}^2(X,Y),
\end{equation}
by Fatou's lemma. Applying the Hewitt--Savage theorem to the right-hand side of the inequality yields
\begin{equation}\label{mean energy}
\lim_{j\to\infty} \frac{\caE_{N_j}(\mu_{{N_j},\beta})}{{N_j}} \ge \int_{\caE(\Omega)} \mathfrak{e}(\nu) \diff P(\nu),
\end{equation}
where
\begin{equation*}
\mathfrak{e}(\nu) := \frac{1}{2}\int_{\Omega\times\Omega} v(X,Y) \diff\nu(X)\diff\nu(Y).
\end{equation*}

We now turn to the entropy density, decomposed as in~(\ref{subadditivity}). For any absolutely continuous measure $\nu$,
\begin{equation*}
\caS_{N-k \lfloor\frac{N}{k}\rfloor }\big(\mu^{N-k \lfloor\frac{N}{k}\rfloor}_{N,\beta}\big)
= \caS_{N-k \lfloor\frac{N}{k}\rfloor }\big(\mu^{N-k \lfloor\frac{N}{k}\rfloor}_{N,\beta} \Vert \nu^{\otimes(N-k \lfloor\frac{N}{k}\rfloor)}\big)
- \int_{\Omega^{N-k \lfloor\frac{N}{k}\rfloor}} \mu^{N-k \lfloor\frac{N}{k}\rfloor}_{N,\beta}\ln \nu^{\otimes (N-k \lfloor\frac{N}{k}\rfloor)} \diff X_1^N,
\end{equation*}
where the first term is the relative entropy of $\mu^{N-k \lfloor\frac{N}{k}\rfloor}_{N,\beta}$ with respect to $ \nu^{\otimes(N-k \lfloor\frac{N}{k}\rfloor)}$, which is non-positive. As for the other term, the symmetry of the potential implies that
\begin{align*}
\frac{1}{N}\int_{\Omega^{N-k \lfloor\frac{N}{k}\rfloor}} \mu^{N-k \lfloor\frac{N}{k}\rfloor}_{N,\beta}\ln \nu^{\otimes (N-k \lfloor\frac{N}{k}\rfloor)} \diff X_1^N
= \frac{1}{N} \left(N-k\left\lfloor\frac{N}{k}\right\rfloor\right) \int_{\Omega} \mu^{1}_{N,\beta}(X)\ln \nu(X)\diff X,
\end{align*}
 and the choice $\nu(X) = \vert \Omega \vert^{-1}$ yields that this term vanishes in the limit $N\to\infty$ since $N-k\left\lfloor\frac{N}{k}\right\rfloor\le k$. Hence, (\ref{subadditivity}) yields
\begin{equation*}
\liminf_{j\to\infty} -\frac{\caS_{N_j}(\mu_{N_j,\beta})}{N_j} 
\ge -\liminf_{j\to\infty} \frac{1}{N_j}\left\lfloor\frac{N_j}{k}\right\rfloor \caS_k(\mu_{N_j,\beta}^k)
\ge -\frac{\caS_k(\mu_{\infty,\beta}^k)}{k}
\end{equation*}
for all $k\in\bbN$, by lower semi-continuity of minus the entropy. In particular $\liminf_{j\to\infty} -{N_j}^{-1}\caS_{N_j}(\mu_{N_j,\beta})\ge \liminf_{k\to\infty} - k^{-1}\caS_k(\mu_{\infty,\beta}^k)$.

Now, for any two probability measures $\rho^k,\nu^k$, 
\begin{align*}
\lambda \caS_k(\rho^k) + (1-\lambda)\caS_k(\nu^k) &\le \caS_k\left(\lambda \rho^k + (1-\lambda)\mu^k\right) \\
&\le  \lambda\caS_k(\rho^k) + (1-\lambda)\caS_k(\mu^k) - \lambda\ln\lambda - (1-\lambda)\ln(1-\lambda)\\
&\le \lambda\caS_k(\rho^k) + (1-\lambda)\caS_k(\mu^k)  + \ln 2,
\end{align*}
where we used first the concavity of the entropy and then the fact that the logarithm is an increasing function. Dividing these inequalities by $k$ and taking the limit $k\to\infty$, we have shown that the map
\begin{equation*}
\mu_{\infty,\beta}\longmapsto \liminf_{k\to\infty} -\frac{\caS_k(\mu_{\infty,\beta}^k)}{k}
\end{equation*}
is affine. Hence, 
\begin{align}
\liminf_{k\to\infty} -\frac{\caS_k(\mu_{\infty,\beta}^k)}{k} 
&= \liminf_{k\to\infty} -k^{-1}\caS_k \left[\int_{\caE(\Omega)}\nu^{\otimes k}\diff P(\nu)\right] \nonumber \\
&= \int_{\caE(\Omega)} \liminf_{k\to\infty} -k^{-1}\caS_k\left(\nu^{\otimes k}\right)\diff P(\nu) 
= \int_{\caE(\Omega)} \caS_1\left(\nu\right)\diff P(\nu), \label{mean entropy}
\end{align}
where the first equality is a consequence of the Hewitt--Savage theorem. Note that in both (\ref{mean energy}) and (\ref{mean entropy}), the measure $P$ depends on the chosen subsequence. 

This and the variational upper bound~(\ref{MF: upper}) yield
\begin{equation*}
\lim_{j\to\infty}\frac{F_{N_j,\beta}}{N_j} = \lim_{j\to\infty}\frac{\caF_{N_j,\beta}(\mu_{{N_j},\beta})}{N_j} = \int_{\caE(\Omega)} \mathfrak{f}_\beta\left(\nu\right)\diff P(\nu).
\end{equation*}
Finally, assume by contradiction that the support of $P$ contains measures that are not minimisers of $\mathfrak{f}_{\beta}$, then $\liminf_{N\to\infty} N^{-1}\caF_{N,\beta}(\mu_{N,\beta})>f_\beta$, a contradiction with the initial remark of this proof. This concludes the proof, with $P=P_\beta$ supported on the minimisers of $\mathfrak{f}_\beta$, and hence 
$\lim_{N\to\infty} N^{-1}F_{N,\beta} = f_\beta$.

It remains to show that any accumulation point is absolutely continuous. For this, it suffices to show that, for any fixed $1\le k\le N$,
the densities $\rho_{N,\beta}^k(X_1^k)$ are bounded, uniformly in $N$.  We decompose 
the interaction potential as
\begin{multline*}
V_{N}(X_1^N) = V_{N}^{\mathrm{(int)}}(X_1^k) + V_{N}^{\mathrm{(cross)}}(X_1^N) 
+ V_{N}^{\mathrm{(ext)}}(X_{k+1}^N) \\
:= \frac{1}{N-1} \Big[\sum_{1\le i<j\le k} + \sum_{1\le i\le k}\sum_{j > k} 
+ \sum_{k<i<j\le N}\Big] v\left(\abs{x_i-x_j}; m_i,m_j\right),
\end{multline*}
and use $\exp\big[-\beta\big(V_{N}^{\mathrm{(int)}}(X_{1}^k) + V_{N}^{\mathrm{(cross)}}(X_{1}^N)\big)\big]\le 1$ 
to get the immediate bound
\begin{equation}\label{rho_bound}
\rho_{N,\beta}^k(X_1^k) \le \frac{Z_{N}^{\mathrm{(ext)}}(\beta)}{Z_{N}(\beta)},
\end{equation}
where 
$$
Z_{N}^{\mathrm{(ext)}}(\beta) = \int_{\Omega^{N-k}}\exp\big(-\beta V_{N}^{\mathrm{(ext)}}(X_{k+1}^N)\big)\diff X_{k+1}^N.
$$
Applying Jensen's inequality on $\Omega^N$ with the probability measure 
$$\diff\Theta(X_1^N):=
\big(\vert \Omega\vert^k Z_{N}^{\mathrm{(ext)}}(\beta)\big)^{-1} 
\exp\big(-\beta V_{N}^{\mathrm{(ext)}}(X_{k+1}^N)\big)\diff X_{1}^N
$$ 
yields
\begin{align}\label{Jensen}
Z_{N}(\beta) &= \vert \Omega\vert^k Z_{N}^{\mathrm{(ext)}}(\beta)
\int_{\Omega^N} \exp\left(-\beta \Big(V_{N}^{\mathrm{(int)}} 
+ V_{N}^{\mathrm{(cross)}}\Big) \right) \diff\Theta(X_1^N) \notag \\
& \ge \vert \Omega\vert^k Z_{N}^{\mathrm{(ext)}}(\beta)
\exp\left(-\beta  \int_{\Omega^N} \Big(V_{N}^{\mathrm{(int)}}+V_{N}^{\mathrm{(cross)}}\Big) \diff\Theta(X_1^N) \right).
\end{align}
Observe now that, in view of Assumption~\ref{A:VBasic}~(c),
\begin{align*}
\int_\Omega v(X,Y)\diff X 
&= \int_{\Lambda\times{M}} v(|x-y|;m,m')\diff x \diff m \\
&\le \omega_d \int_M \int_0^D r^{d-1} v(r;m,m') \diff r  \diff m\\
&\le |M|\Vert J\Vert_{L^\infty(M\times M)},
\end{align*}
for all $ Y=(y,m')\in \Omega$. Hence, it follows by Assumption~\ref{A:VBasic}~(a) that
\begin{align}\label{Vin_est}
\int_{\Omega^N} V_{N}^{\mathrm{(int)}}(X_1^k)\diff\Theta(X_1^N)
&=|\Omega|^{-k}\int_{\Omega^k} V_{N}^{\mathrm{(int)}}(X_1^k) \diff X_1^k \notag\\
&=|\Omega|^{-k}\frac{k(k-1)}{2(N-1)}|\Omega|^{k-2}\int_{\Omega^2}v(X,Y)\diff X\diff Y \notag\\
&\le \frac{k(k-1)|M|\Vert J\Vert_{L^\infty}}{2(N-1)|\Omega|} \notag\\
&\le k|\Lambda|^{-1}\Vert J\Vert_{L^\infty}
\end{align}
and that
\begin{align}\label{Vcross_est}
\int_{\Omega^N} V_{N}^{\mathrm{(cross)}}(X_1^N)\diff\Theta(X_1^N) 
&= \frac{k(N-k)}{N-1}\int_{\Omega^{N-1}} \left(\int_\Omega v(X_1,X_{k+1})\diff X_1\right)\diff\Theta(X_2^N) \nonumber\\
&\le \frac{k(N-k)|M|\Vert J\Vert_{L^\infty}}{(N-1)|\Omega|} \notag\\
&\le k|\Lambda|^{-1}\Vert J\Vert_{L^\infty}.
\end{align}
Now, by \eqref{rho_bound}, \eqref{Jensen}, \eqref{Vin_est} and \eqref{Vcross_est}, 
$$
\rho_{N,\beta}^k(X_1^k) \le |\Omega|^{-k} \exp\big(2k|\Lambda|^{-1}\Vert J\Vert_{L^\infty}\big), 
\quad\text{for all} \ 1\le k\le N,
$$
which concludes the proof. 
\end{proofof} 


\subsection{Self-consistency equation}
We now prove that the minimisers of the one-particle functional obtained in Theorem~\ref{effective theory}
are solutions of the self-consistency equation (SCE) \eqref{GapEq}, as 
claimed in Proposition~\ref{prop:Gap Eq}.

\begin{proofof} {\em Proposition~\ref{prop:Gap Eq}.} 
We denote by $L^1(\Omega)_+$ the cone of almost everywhere positive functions
$\nu\in L^1(\Omega)$. With the slight abuse of notation $\dif\nu(X)=\nu(X)\dif X$,
we seek critical points in $L^1(\Omega)_+$ of the functional $\mathfrak{f}_\beta$ 
under the constraint $\int_\Omega \nu(X)\dif X=1$. By the Lagrange multiplier theorem, 
if $\nu_0\in L^1(\Omega)_+$ is a constrained critical point, there exists $\lambda\in\real$ such that
$$
\int_{\Omega} v(x,x';m,m') \nu_0(x,m')\diff x'\diff m' + \beta^{-1} (1 + \ln\nu_0(x,m)) - \lambda=0,
\quad \text{a.e.} \ (x,m)\in\Omega.
$$ 
In view of the normalisation $\int_\Omega \nu_0=1$, 
the proposition now follows by taking exponentials of both sides of this equation.
\end{proofof}


\subsection{Phase transitions}

In this subsection we study in detail the structure of the set of solutions 
of the homogeneous SCE \eqref{gap_hom}, depending on the inverse temperature $\beta$. 
Our proofs rely on standard bifurcation theory, mainly the {\em Crandall--Rabinowitz theorem}
which describes, for general nonlinear problems, bifurcation from a simple eigenvalue of the
linearisation. A simple exposition can be found in Chapter~2 of \cite{am}. The
theory itself relies on linear functional analysis, in particular on properties of {\em Fredholm operators},
which are recalled in \cite{am} and discussed in more detail e.g.~in Chapter~6 of \cite{Brezis:2011aa}.

\subsubsection{High temperature and transcritical bifurcation}

In order to address the transcritical bifurcation, we first need to
compute the relevant derivatives of the mapping $\Phi$ defined in
\eqref{Phi}. We denote by $D_\beta \Phi(\beta,\nu)$ 
the Fr\'echet derivative of $\Phi$ with respect to
$\beta$, evaluated at the point $(\beta,\nu)\in (0,\infty)\times L^2(M)$, and we use
a similar notation for the other derivatives of $\Phi$.
The following lemma can then be proved by routine verifications.

\begin{lem}\label{derivatives.lem}
The mapping $\Phi:(0,\infty)\times L^2(M)\to L^2(M)$ is smooth and the 
following formulas hold: 
\begin{equation*}
D_\beta \Phi(\beta,\nu)=\left[H_\nu-
\frac{\int_{M} H_\nu\e^{-\beta H_\nu}}{\int_{M} \e^{-\beta H_\nu}}\right]
\frac{\e^{-\beta H_\nu}}{\int_{M} \e^{-\beta H_\nu}};
\end{equation*}
\begin{equation*}
D_\nu \Phi(\beta,\nu)\mu=\mu+\beta
\left[H_\mu-\frac{\int_{M} H_\mu\e^{-\beta H_\nu}}{\int_{M} \e^{-\beta H_\nu}}\right]
\frac{\e^{-\beta H_\nu}}{\int_{M} \e^{-\beta H_\nu}}, \quad \mu\in L^2(M);
\end{equation*}
\begin{multline*}
D^2_{\beta \nu} \Phi(\beta,\nu)\mu=
\left\{1-\beta\left[H_\nu-\frac{\int_{M} H_\nu\e^{-\beta H_\nu}}{\int_{M} \e^{-\beta H_\nu}}\right]\right\}
\left[H_\mu-\frac{\int_{M} H_\mu\e^{-\beta H_\nu}}{\int_{M} \e^{-\beta H_\nu}}\right]
\frac{\e^{-\beta H_\nu}}{\int_{M} \e^{-\beta H_\nu}}\\
-\beta\frac{\int_{M} H_\nu\e^{-\beta H_\nu}\int_{M} H_\mu\e^{-\beta H_\nu}
-\int_{M} \e^{-\beta H_\nu}\int_{M} H_\nu H_\mu\e^{-\beta H_\nu}}
{\left(\int_{M} \e^{-\beta H_\nu}\right)^2}
\frac{\e^{-\beta H_\nu}}{\int_{M} \e^{-\beta H_\nu}}, \quad \mu\in L^2({M});
\end{multline*}
\begin{multline*}
D^2_{\nu\nu} \Phi(\beta,\nu)[\mu,\rho]=
\beta^2\left[
\frac{\int_{M} H_\mu H_\rho\e^{-\beta H_\nu}+H_\mu\int_{M} H_\rho\e^{-\beta H_\nu}
+H_\rho\int_{M} H_\mu\e^{-\beta H_\nu}}{\int_{M} \e^{-\beta H_\nu}}
\right]\frac{\e^{-\beta H_\nu}}{\int_{M} \e^{-\beta H_\nu}} \\
-\beta^2\left[2\frac{\int_{M} H_\rho\e^{-\beta H_\nu}\int_{M} H_\mu\e^{-\beta H_\nu}}
{\left(\int_{M} \e^{-\beta H_\nu}\right)^2}
+H_\rho H_\mu\right]\frac{\e^{-\beta H_\nu}}{\int_{M} \e^{-\beta H_\nu}},
\quad \rho,\mu\in L^2({M}).
\end{multline*}
\end{lem}
It follows from Lemma~\ref{derivatives.lem} that
\begin{equation}\label{compact_pert}
D_\nu \Phi(\beta,|{M}|^{-1})= I-\beta K,
\end{equation}
and
\begin{equation}\label{formula_for_K}
D^2_{\beta \nu} \Phi(\beta_\star,|M|^{-1})\mu(m)=
|{M}|^{-1}\int_M \vt(m,m')\mu(m')\diff m'=-(K\mu)(m),
\end{equation}
where the operator $K:L^2(M)\to L^2(M)$ has been defined in \eqref{K}.


\medskip
\begin{proofof} {\em Proposition~\ref{high_temp.prop}.}
If $\beta^{-1}>\norm{K}$ then the right-hand side of \eqref{linearisation} is a contraction,
and so \eqref{linearisation} has only the trivial solution $\mu=0$ 
by the contraction mapping principle. Since the operator $K$ is self-adjoint and compact,
it follows from the Fredholm alternative \cite[Theorem~6.6]{Brezis:2011aa} that the linear operator 
$D_\nu \Phi(\beta,|{M}|^{-1}):L^2(M)\to L^2(M)$ is 
an isomorphism, for any $\beta>0$. Therefore, invoking the implicit function theorem,
we conclude that through each solution $(\beta,|{M}|^{-1})$ on the isotropic
line, there passes a unique (local) curve of solutions of \eqref{operator_gap}, that is,
the isotropic line itself. Hence there can be no bifurcations.
\end{proofof}


\begin{proofof} {\em Theorem~\ref{transition.thm}.}
The result is a direct consequence of
the Crandall--Rabinowitz theorem (see \cite{cranrab} or \cite[Theorem~2.8]{am})
which, in the present context, can be stated as follows.
\begin{thm}\label{bifurcation.thm}
For $\beta>0$, we let $A_{\beta}:=D_\nu\Phi(\beta,|{M}|^{-1}): L^2({M})\to L^2({M})$,
and we make the following assumptions:
\begin{itemize}
\item[(A)] there exists $\beta_\star>0$ such that
 $\ker A_{\beta_\star}=\vect\{\mu_\star\}$ for some $\mu_\star\in L^2({M})\sm\{0\}$;
\item[(B)] $A_{\beta_\star}$  is a Fredholm operator of index zero;
\item[(C)] $D^2_{\beta \nu} \Phi(\beta_\star,|{M}|^{-1}) \mu_\star \not\in \rge A_{\beta_\star}$.
\end{itemize}
Then there exists $\ep>0$ and a continuous map 
\begin{equation}\label{curve}
s\mapsto (\beta(s),\rho(s)): (-\ep,\ep)\to (0,\infty)\times[\mu_\star]^\perp
\end{equation}
with the following properties:
\begin{itemize}
\item[(a)] $(\beta(s),\nu(s))$ is a solution of \eqref{operator_gap} for all $s \in(-\ep,\ep)$, 
where $\nu(s):=|M|^{-1}+s(\mu_\star+\rho(s))$;
\item[(b)] $\beta(0)=\beta_\star$ and $\rho(0)=0$;
\item[(c)] there is a neighbourhood $\mathcal{O}$ of $(\beta_0,|M|^{-1})$ in  
$(0,\infty)\times L^2(M)$ such that all solutions of \eqref{operator_gap} inside $\mathcal{O}$
belong to the curve \eqref{curve}.
\end{itemize}
\end{thm}
The crucial step in applying Theorem~\ref{bifurcation.thm}
is to check (A), i.e.~to find the transition temperature
$\beta_\star$ where the bifurcation occurs. Then the other 
assumptions are automatically satisfied, as ensured by the following lemma.

\begin{lem}\label{BC.lem}
If $A_{\beta_\star}$ satisfies (A) for some $\beta_\star>0$, then (B) and (C) also hold.
\end{lem} 

\begin{proof}
It follows immediately from \eqref{compact_pert} that
\begin{itemize}
\item[(i)] $A_{\beta_\star}$ is a bounded self-adjoint operator;
\item[(ii)] $A_{\beta_\star}$ is a compact perturbation of the identity.
\end{itemize}
The Fredholmness of $A_{\beta_\star}$ is a consequence of (ii). Then $A_{\beta_\star}$ has a closed
range, and the zero index property follows
from (i) by the usual decomposition $L^2(M)=\ker A_{\beta_\star} \oplus \rge A_{\beta_\star}$. 
This proves (B).

To prove (C), suppose by contradiction that there exists $\lambda\in L^2(M)$ such that
\begin{equation}\label{transverse1}
D^2_{\beta \nu} \Phi(\beta_\star,|M|^{-1}) \mu_\star=D_\nu \Phi(\beta_\star,|M|^{-1}) \lambda.
\end{equation}
By \eqref{formula_for_K}, this is equivalent to
$$
-K\mu_\star=(I-\beta_\star K)\lambda.
$$
Using (A), this becomes 
$$
(I-\beta_\star K)\lambda=-\beta_\star^{-1}\mu_\star.
$$
By the Fredholm alternative \cite[Theorem~6.6]{Brezis:2011aa}, 
this equation has a solution $\lambda$ if and only if
$\mu_\star\in (\ker A_{\beta_\star})^\perp$, a contradiction. Therefore, (C) must hold.
\end{proof}
Combining Theorem~\ref{bifurcation.thm} and Lemma~\ref{BC.lem} proves
Theorem~\ref{transition.thm}.
\end{proofof}


\begin{proofof} {\em Proposition~\ref{kernel.prop}.}
To compute the kernel it is convenient to expand the potential in
spherical harmonics. Consider any pair of directions $m,m'\in\stwo$. For
$\gamma=\measuredangle(m,m')$, there holds
\begin{equation}\label{addition}
 P_2(\cos\gamma)=\frac{4\pi}{5}
 \sum_{m=-2}^2 Y_{2 m}(\theta,\varphi)Y^*_{2 m}(\theta',\varphi'),
\end{equation}
where $(\theta,\varphi)$ and $(\theta',\varphi')$ are the respective spherical coordinates 
of $m$ and $m'$. We use the standard notation for spherical harmonics
on $\stwo$ and ${\cdot}^*$ denotes complex conjugation.

Now, we are looking for $\beta>0$ and $\mu=\mu(\theta,\varphi)$ 
satisfying the linearised equation $\mu=\beta K\mu$. For the purpose of calculations, we extend
$\mu$ by antipodal symmetry, $\mu(m)\equiv\mu(-m)$, to the whole of $\stwo$. We can then decompose
it as
\begin{equation}\label{decomp_mu}
 \mu(\theta,\varphi)=
 \sum_{\ell=0}^\infty\sum_{m=-\ell}^\ell \mu_{\ell m}Y_{\ell m}(\theta,\varphi),
\end{equation}
where the (complex) coefficients $\mu_{\ell m}$ are given by
\begin{equation*}
 \mu_{\ell m}=\int_0^{2\pi}\int_0^{\pi} 
 \mu(\theta,\varphi)Y^*_{\ell m}(\theta,\varphi) \sin\theta\diff\theta\diff\varphi.
\end{equation*}

On the other hand, it follows from \eqref{constancy} that 
$\vt(m,m')=-wP_2(\cos\gamma)$
in \eqref{dev}. Using the invariance of $\mu$ and of the spherical harmonics under the antipodal transformation $m\mapsto{-m}, \ m\in\stwo$, we then find by \eqref{addition} that 
the operator $K$ in \eqref{K} has the explicit form
$$
(K\mu)(\theta,\ffi)
=\frac{\ow}{5}\sum_{m=-2}^2\mu_{2 m} Y_{2 m}(\theta,\varphi), \quad \mu\in L^2(M).
$$
Therefore, using \eqref{decomp_mu}, the linearised equation $\mu=\beta K\mu$ becomes
\begin{equation*}
 \sum_{\ell=0}^\infty\sum_{m=-\ell}^\ell \mu_{\ell m}Y_{\ell m}(\theta,\varphi)
 =\frac{\beta\ow}{5}\sum_{m=-2}^2\mu_{2 m} Y_{2 m}(\theta,\varphi).
\end{equation*}
Identifying the respective coefficients of the spherical harmonics then yields
\begin{equation*}
  \mu_{\ell m}=0 \quad \forall \ell\neq2, \qquad  
  \mu_{2 m}=\frac{\beta\ow}{5}\mu_{2 m} \quad \forall \,m\in\{-2,-1,0,1,2\}.
\end{equation*}
Hence the kernel of $A_\beta$ is non-trivial if and only if 
\begin{equation*}
 \beta=\beta_\star(\ow):=\frac{5}{\ow}.
\end{equation*}
Furthermore, seeking a real solution yields $\mu_{2 m}=0$ for $m\neq0$, so that
$\ker A_{\beta_\star}=\vect\{\mu_\star\}$ with
\begin{equation*}
 \mu_\star(\theta,\varphi)= \mu_\star(\theta)=3 \cos^2\theta-1.
\end{equation*}

Finally, in view of the discussion in Remark~2.9 of \cite{am},  
the transcriticality of the bifurcation follows from the observation that
\begin{equation}\label{transcriticality}
B:=\la \mu_\star,D^2_{\nu\nu}\Phi(\beta_\star,|M|^{-1})[\mu_\star,\mu_\star] \ra \neq0.
\end{equation}
To prove \eqref{transcriticality}, 
we use the formula for the second derivative given in Lemma~\ref{derivatives.lem}.
Since $\int_M H_{\mu_\star}\diff m=0$, we find
$$
D^2_{\nu\nu}\Phi(\beta_\star,|M|^{-1})[\mu_\star,\mu_\star]
=\beta_\star^2|M|^{-2}\left[\int_M H_{\mu_\star}^2\diff m - |M|H_{\mu_\star}^2\right].
$$
Moreover $\int_M \mu_\star\diff m=0$, and so
\begin{align*}
B = -\beta_\star^2|M|^{-1} \la\mu_\star,H_{\mu_\star}^2\ra
= -\beta_\star^2|M|^{-1}\int_{M}
\mu_\star(m)\left(\int_M U(m,m')\mu_\star(m')\diff m'\right)^2\diff m.
\end{align*}
Now an explicit calculation using \eqref{MSpotential}, \eqref{eigen} 
and the addition formula \eqref{addition} yields 
$$
\intstwo U(m,m')\mu_\star(m')\diff m'
=-\ow\frac{2\pi}{5}(3\cos^2\theta-1)
$$
and, finally,
$$
B=\frac{8}{5}\Big(\frac{2\pi}{5}\Big)^2\beta_\star^2\ow^2>0.
$$
The proposition is proved.
\end{proofof}


\subsubsection{Low temperature}

Finally, we analyse in some detail the zero temperature limit, where one expects the thermal state to describe all molecules to be perfectly aligned. The natural framework in this case is to consider $\beta$
in the extended half-line $(0,\infty]$ and $\nu\in F$, where $F$ is the Banach space of regular signed measures over $M$, 
equipped with the norm of total variation. By the Riesz--Markov theorem, 
the norm of $\mu\in F$ is given by
\begin{equation}\label{Riesz-Markov}
\Vert \mu\Vert = |\mu|(M) = \sup \{\vert\mu(f)\vert: f\in C^0(M), \ \Vert f\Vert_{L^\infty} = 1\},
\end{equation}
where $|\mu|$ is the total variation of $\mu$. We shall use the standard notation $\mu(f) := \int_{M} f \diff\mu$. Note that probability measures lie on the unit sphere of $F$. 
\begin{lem}\label{lem:HConv}
For any $\nu\in F$, let
\begin{equation*}
H_\nu(m) := \nu(U(m,\cdot)) = \int_{M} U(m,m')\diff\nu(m').
\end{equation*}
\begin{enumerate}
\item If $U\in C^k(M\times M;\bbR)$, then $H_\nu\in C^k(M;\bbR)$.
\item If, furthermore, $\nu\rightharpoonup\mu$, then $\partial^\alpha H_\nu\to \partial^\alpha H_\mu$ pointwise for any $\vert\alpha\vert \le k$.
\item If $\nu\to\mu$ in $F$, then $\partial^\alpha H_\nu\to \partial^\alpha H_\mu$ uniformly for any $\vert\alpha\vert \le k$.
\end{enumerate}
\end{lem}
\begin{proof}
(a) Since $M$ is compact, $U\in C^k(M\times M;\bbR)$ implies that $\sup_{m,m'\in M\times M}\vert \partial_m^\alpha U(m,m')\vert <\infty$ for any $\vert\alpha\vert \le k$. Furthermore, $\nu(M) \le \Vert \nu\Vert$, so by dominated convergence $\partial^\alpha H_\nu(m) = \int_{M}\partial^\alpha_m U(m,m')\diff\nu(m')$ is a continuous function for $\vert\alpha\vert \le k$. \\
\noindent(b) This follows by definition of the weak convergence of measures applied to the function $\partial_m^\alpha U(m,\cdot)$ with fixed $m$. \\
\noindent(c) If $\nu$ converges in norm to $\mu$, then in particular $\sup_{m\in M}\vert \nu(\partial_m^\alpha U(m,\cdot)) - \mu(\partial_m^\alpha U(m,\cdot)) \vert\to 0$, which is the uniform convergence of $\partial^\alpha H_\nu$ to $\partial^\alpha  H_\mu$.\label{neighb}
\end{proof}

For any $\beta\in\bbR$,
\begin{equation*}
\tau_\beta^\nu(f) := \frac{\int_{M} f(m) \e^{-\beta H_\nu(m) }\diff{m}}{\int_{M} \e^{-\beta H_\nu(m) }\diff{m}},\quad f\in C^0(M),
\end{equation*}
defines a probability measure $\tau_\beta^\nu\in F$. The map $\Phi(\beta,\nu) = \nu - \tau_\beta^\nu$ introduced in \eqref{Phi} is now interpreted as $\Phi:(0,\infty]\times F\to F$, by extending it as
\begin{equation*}
\Phi(\infty,\nu):= \nu - \delta_{m_0},
\end{equation*}
where $\delta_{m_0}$ is the Dirac mass at $m_0$, i.e.~$\delta_{m_0}(f)=f(m_0)$ for any
continuous function $f$.

We shall now proceed under the general assumption of an \emph{axially symmetric} potential, namely 
$$U(m,m') = U(\cos^2\gamma),$$ 
where $\gamma$ is the angle between $m$ and $m'$, having 
a unique, strict global minimum, corresponding to $\gamma = 0$.
It follows that $H_{\delta_{m_0}}(m) = U(m,m_0)$ has a unique, strict minimum 
at $m= m_0$. We shall further suppose that 
there exists a neighbourhood $\caD$ of $\delta_{m_0}$ in $F$ such that 
$H_\nu(m)$ has a unique minimum at $m= m_0$.
It is only at the very end of this section, in the proof of Theorem~\ref{zeroK.thm}, that we restrict our attention to the concrete case of the Maier--Saupe interaction~\eqref{effpot}, which trivially satisfies the above assumptions.

With these preliminaries, we now construct a continuous branch of axisymmetric solutions 
emanating from the perfectly aligned state at zero temperature, $(\infty,\delta_{m_0})$. 
The point $m_0$ corresponds to the `director' in the liquid crystal jargon. 
By the full rotation symmetry of the problem, all $m_0$ are equivalent and we may as well choose it to be the north pole. 
We shall therefore use coordinates $(\theta,\varphi)$ such that $m_0$ corresponds to $\theta = 0$ in what follows. 
By an {\em axisymmetric solution}, we then mean a solution $(\beta,\nu)$ of $\Phi(\beta,\nu)=0$ such that $\nu$ 
does not depend on $\varphi$. 
This does not exclude the breaking of the axial symmetry. However, the problem being invariant under rotations about the director, if the symmetry were broken, 
the branch we construct would be given by an average of the $\varphi$-dependent solutions.
Note that for the Maier--Saupe potential, \cite{slast} actually proves that all solutions are indeed axisymmetric. 
We shall henceforth write $\theta$ instead of $m$ and $0$ instead $m_0$. 
We will correspondingly abuse the notation, using the same symbol for objects depending on $m$; for instance, 
$f(m)\equiv f(\theta)$, etc. As usual, we may also simply write $\delta$ for $\delta_0$.

The following proposition, which holds for an arbitrary potential $U$ in the class discussed above, 
provides a description of the asymptotic states in the low temperature limit. 
It is based on Laplace's method, as presented in Lemma~\ref{asympt.lem} in the Appendix.

\begin{prop}\label{prop:PhiAsymp}
Assume that $U\in C^\infty(M\times M)$. Then the following holds.
\begin{enumerate}
\item Consider $\nu\in F$ such that $H_\nu''(0) >0$. Then
\begin{equation*}
\Phi(\beta,\nu) \rightharpoonup \nu - \delta
\end{equation*}
as $\beta\to\infty$, in the sense of weak convergence of measures.

\item For any $g\in C^2(M)$, consider $\Phi_g: (0,\infty]\times F\to\bbR$ defined by $\Phi_g(\beta,\nu) = \Phi(\beta,\nu)(g)$. If 
$$\partial_{\theta\theta}^2U(\theta,\theta')\vert_{\theta = \theta' = 0}>0,$$ 
there is a neighbourhood $\caU\subset F$ of $\delta$ such that 
$\Phi_g\in C^0((0,\infty]\times\caU;\bbR)\cap C^1((0,\infty)\times\caU;\bbR)$ with
\begin{align*}\label{derivatives_phig}
D_\beta \Phi_g(\beta,\nu)&=\left(\tau_\beta^\nu(g H_\nu) - \tau_\beta^\nu(g)\tau_\beta^\nu(H_\nu)\right),\\
D_\nu \Phi_g(\beta,\nu)[\mu]&= \mu(g)+\beta\left(\tau_\beta^\nu(g H_\mu) - \tau_\beta^\nu(g)\tau_\beta^\nu(H_\mu)\right), \ \mu\in F.
\end{align*}
Furthermore,
$\Phi_g$ is differentiable with respect to $\nu$ at $(\infty,\nu)$ for all $\nu\in\caU$, with
$$
D_\nu \Phi_g(\infty,\nu)[\mu] = \mu(g), \ \mu\in F.
$$
\end{enumerate}
\end{prop}

The proof relies on several lemmas. Let us first introduce a notation: for any $g\in C^0 ([0,\pi/2];\mathbb{R})$, let
\begin{equation*}
\langle g \rangle_\beta^f := \frac{\int_{0}^{\pi/2}\e^{-\beta f(\theta) }g(\theta)\sin\theta\diff{\theta}}{\int_{0}^{\pi/2}\e^{-\beta f(\theta) }\sin\theta\diff{\theta}}.
\end{equation*}

\begin{lem}\label{lem:cumulant}
For any $g,h\in C^2([0,\pi/2];\mathbb{R})$,
\begin{equation*}
\lim_{\beta\to\infty}\beta^{1/2} \left(\langle g h \rangle_\beta^f - \langle g \rangle_\beta^f  \langle h \rangle_\beta^f \right) = 0.
\end{equation*}
If, moreover, $h'(0) = 0$, then
\begin{equation*}
\lim_{\beta\to\infty}\beta \left(\langle g h \rangle_\beta^f - \langle g \rangle_\beta^f  \langle h \rangle_\beta^f \right) = 0.
\end{equation*}
\end{lem}
\begin{proof}
It suffices to apply~(\ref{LaplaceE}) to $\langle g h \rangle_\beta^f - \langle g \rangle_\beta^f  \langle h \rangle_\beta^f $ and gather terms of same order. For simplicity, we write $\langle g \rangle_\beta^f = g(0) + A g'(0) \beta^{-1/2} + (B_1 g''(0) + B_2 g'(0))\beta^{-1}$, with the obvious definition of
$A,B_1,B_2$.
\begin{align*}
\text{Order }\beta^0:\quad&(gh)(0) - g(0) h(0) \\
\text{Order }\beta^{-1/2}:\quad&A\left[(gh)'(0) - (g'(0) h(0) + g(0) h'(0))\right] \\
\text{Order }\beta^{-1}:\quad& (B_1(gh)''(0) + B_2(gh)'(0)) - A^2 g'(0)h'(0) \\ & - B_1 (g(0)h''(0) + g''(0) h(0)) - B_2(g(0)h'(0) + g'(0) h(0))
\end{align*}
The first two orders vanish for any continuously differentiable functions $g,h$, while the order $\beta^{-1}$ requires the additional condition $h'(0) = 0$ to vanish.
\end{proof}
Finally, we also recall the following lemma.
\begin{lem}[\!\!\cite{Simon:1979wt}]\label{simon.lem}
Let $K$ be a compact set and $\mu_0$ a probability measure on $K$. 
For $h\in C^0(K)$, let $\diff\mu_h:= \exp(h)\diff\mu_0 / \int_K\exp(h)\diff\mu_0$. 
Then $\Vert \mu_h - \mu_g \Vert \le \Vert h-g\Vert_{L^\infty}$, for all $h,g\in C^0(K)$.
\end{lem}

We are now in a position to give the

\begin{proof}[Proof of Proposition~\ref{prop:PhiAsymp}]
(a) Let $g\in C^0(M)$. It suffices to note that $\tau_\beta^\nu(g) = \langle g \rangle_\beta^{H_\nu}$ in the notation of Lemma~\ref{asympt.lem}, so that $\lim_{\beta\to\infty}(\nu - \tau_\beta^\nu)(g) = \nu(g) - g(0)$.

\noindent (b) Let $\nu_n\to\delta$ in norm. By Lemma~\ref{lem:HConv}~(c), any derivative of $H_{\nu_n}$ converges uniformly to that of $H_{\delta}$, where $H_{\delta}(m) = U(m,m_0)$. 
Let $0 < \epsilon < H_{\delta}''(0)$. There is $n_0$ such that
\begin{equation*}
H_{\nu_n}''(0) \ge H_{\delta}''(0) - \vert H_{\nu_n}''(0) - H_{\delta}''(0) \vert \ge H_{\delta}''(0) - \epsilon >0
\end{equation*}
for all $n\ge n_0$. Hence, there exists a neighbourhood $\caU$ of $\delta$ in $F$ such that 
\begin{equation}\label{lower_bd}
H_{\nu}''(0) \ge \frac{1}{2}H_{\delta}''(0) > 0 \quad \text{for all} \ \nu\in\caU. 
\end{equation}

The remainder of the proof now follows in several steps.

\medskip
\noindent\emph{Continuity at $(\beta,\nu)$ for $\beta<\infty,\nu\in F$.} Let $(\beta_n,\nu_n)\to(\beta,\nu)$. By Lemma~\ref{simon.lem} and again Lemma~\ref{lem:HConv}~(c),
\begin{equation*}
\big\Vert \tau_{\beta_n}^{\nu_n} - \tau_{\beta}^{\nu}\big\Vert
\le \big\Vert \beta_n H_{\nu_n} - \beta H_{\nu}\big\Vert_{L^\infty}\to0 \quad\text{as} \ n\to\infty,
\end{equation*}
since the sequence $(\beta_n)_{n\in\bbN}$ is bounded.

\medskip
\noindent\emph{Continuity at $(\infty,\nu)$ for $\nu\in \caU$.} Let $(\beta_n,\nu_n)\to(\infty,\nu)$, with $\beta_n<\infty$. We proceed as in the proof of Lemma~\ref{asympt.lem}~(b) to compute $\tau_{\beta_n}^{\nu_n}(g)$, using the uniform lower bound \eqref{lower_bd} on $H_{\nu_n}''(0)$ in order to restrict the integrals to the $\nu_n$-independent interval $[0,\beta_n^{-1/4})$, and conclude similarly. If $\beta_n = \infty$ for $n\ge n_0$, then by definition, $\Phi_g(\infty,\nu_n) - \Phi_g(\infty,\nu) = (\nu_n - \nu)(g) \to 0$ indeed.

\medskip
\noindent\emph{Continuous differentiability at $(\beta,\nu)$ for $\beta<\infty,\nu\in \caU$.} We first note that, 
for any $g\in C^0(M)$,
\begin{align*}
D_\beta\tau_\beta^\nu(g) &= -\left(\tau_\beta^\nu(g H_\nu) - \tau_\beta^\nu(g)\tau_\beta^\nu(H_\nu)\right), \\
D_\nu\tau_\beta^\nu[\mu](g) &= -\beta\left(\tau_\beta^\nu(g H_\mu) - \tau_\beta^\nu(g)\tau_\beta^\nu(H_\mu)\right).
\end{align*}
Let $(\beta_n,\nu_n)\to(\beta,\nu)$. We have
\begin{equation*}
\big\vert \tau_{\beta_n}^{\nu_n}(g H_{\nu_n}) - \tau_\beta^\nu(g H_\nu) \big\vert 
\le \Vert g(H_{\nu_n}- H_\nu)\Vert_{L^\infty} + \Vert \tau_{\beta_n}^{\nu_n} - \tau_{\beta}^{\nu} \Vert \Vert gH_\nu\Vert_{L^\infty} \to 0
\end{equation*}
and similarly for $\tau_{\beta_n}^{\nu_n}(g)\tau_{\beta_n}^{\nu_n}(H_{\nu_n})\to\tau_\beta^\nu(g)\tau_\beta^\nu(H_\nu)$, which proves the continuity of both partial derivatives since $(\beta_n)_{n\in\bbN}$ is bounded.

\medskip
\noindent\emph{Continuous differentiability w.r.t.~$\nu\in \caU$ at $(\infty,\nu)$.} Since 
$$
\Phi_g(\infty,\nu + \epsilon\mu) - \Phi_g(\infty,\nu) = \epsilon\mu, 
$$
it follows that $\Phi_g$
is Fr\'echet differentiable w.r.t.~$\nu\in \caU$ at $(\infty,\nu)$, with derivative
$D_\nu \Phi_g(\infty,\nu)[\mu] = \mu(g), \ \mu\in F$. 
To see this, let $(\beta_n,\nu_n)\to(\infty,\nu)$. Using
Lemma~\ref{lem:cumulant} with the uniform convergence of $H^{(\alpha)}_{\nu_n}$, and the fact that 
$H_\mu'(0) = 0$ for all $\mu\in F$ since $0$ is a local minimum for any $H_\mu$, we obtain that, indeed,
$\lim_{n\to\infty}D_\nu\Phi_g(\beta_n,\nu_n)[\mu] = \mu(g), \ \mu\in F$. This completes the proof
of Proposition~\ref{prop:PhiAsymp}.
\end{proof}

We now apply Proposition~\ref{prop:PhiAsymp} to the Maier--Saupe potential
\begin{equation*}
U(m,m') =1 -P_2(\cos(\gamma)).
\end{equation*}
By axial symmetry, $\nu(\theta,\varphi) = \nu(\theta)$, and observing that
$P_{2}(\cos\theta)=\sqrt{\frac{4\pi}{5}}Y_{20}(\theta)$, the addition formula \eqref{addition} yields
\begin{align}
H_\nu(\theta) &=1-\frac{4\pi}{5} \sum_{m=-2}^2 Y_{2m}(\theta,\varphi)
\int_{M} Y^*_{2m}(\theta',\varphi')\diff \nu(\theta') \nonumber \\
&= 1-\sqrt{\frac{4\pi}{5}}Y_{20}(\theta)
\int_{M}\sqrt{\frac{4\pi}{5}} Y_{20}(\theta')\diff \nu(\theta') 
= 1- \left\langle  P_{2}(\cos(\cdot))\right\rangle_{\nu}P_{2}(\cos\theta). \label{HnuMS}
\end{align}
Consider the function $F$ introduced in \eqref{IFTequ}.
Since $F:(0,\infty]\times[0,1]\to\bbR$, our proof of Theorem~\ref{zeroK.thm} relies upon a
very simple case of Proposition~\ref{prop:PhiAsymp}: 
all remaining information about the measure is encoded in the real parameter $\xi$. 
Ultimately, this arises from the fact that the observable $g(\theta) = \sin^2\theta$ giving rise to the order parameter is closely related to the average-field potential $H_\nu(\theta)$ itself, see~(\ref{HnuMS}). In fact, the addition theorem for higher order spherical harmonics would provide higher order potentials for which the analysis below could be applied in a straightforward way. The possibility of constructing 
a branch of solutions for arbitrary observables and potentials, namely at the level of generality of Proposition~\ref{prop:PhiAsymp}, remains an open question.

\begin{proofof} {\it Theorem~\ref{zeroK.thm}}.
Let $g(\theta) = \sin^2\theta$. We consider the bounded linear map $\xi_g: F\to[0,1]$ given by $\xi_g(\nu) = \nu(g)$. By choosing $\nu$ to be point masses at any $\theta\in[0,\frac{\pi}{2}]$, one shows that $\xi_g$ is onto. Now, $G(\beta,\xi_g(\nu)) = \Phi_g(\beta,\nu)$. By Proposition~\ref{prop:PhiAsymp}~(b), there is a neighbourhood $\caU\subset F$ of $\delta$ such that 
$G\circ (\mathrm{id},\xi_g) \in C^0((0,\infty]\times\caU;\bbR)\cap C^1((0,\infty)\times\caU;\bbR)$ and
has a continuous derivative with respect to $\nu$ at the point $(\infty,\delta)$. By the open mapping theorem, there exists $\varepsilon>0$ such that $[0,\varepsilon) = \xi_g(\caU)$. Since, furthermore, $\xi_g\in C^\infty(\caU)$, it follows that 
$G\in C^0((0,\infty]\times[0,\varepsilon))\cap C^1((0,\infty)\times[0,\varepsilon))$ and
has a continuous derivative with respect to $\xi$ at the point $(\infty,0)$. Hence, the theorem follows by applying a version of the implicit function theorem \cite[Theorem~9.3,~p.~230]{loomis}
to $G\upharpoonright (0,\infty]\times[0,\varepsilon)$ at the point 
$(\infty,0)$, where
\begin{equation*}
G(\infty,0) = 0,\qquad \partial_\xi G(\infty,0) = 1.
\end{equation*}
Observe that, even though the resulting function $\bar\xi(\beta)$, which is $C^0((\bar\beta,\infty])$, may not be $C^1$ up to $\beta=\infty$,
the smoothness of $G$ w.r.t.~$\beta<\infty$ ensures that $\xi\in C^1((\bar\beta,\infty))$ indeed.
The last statement of the theorem follows from the fact that $[0,\bar\delta)\subset [0,\varepsilon)$ belongs to the range of $\xi_g$.
\end{proofof}


\appendix

\section{Laplace's method}\label{appendix}

\begin{lem}\label{asympt.lem}
Let $f\in C^3 ([0,\pi/2];\mathbb{R})$ be such that $0$ is its unique global minimum.  Assume that $f''(0)\neq 0$. Then
the following holds.
\begin{itemize}
\item[(a)] As $\beta\to\infty$,
\begin{equation}\label{laplace1}
\int_{0}^{\pi/2}\e^{-\beta f(\theta) }\sin\theta\diff{\theta}
=\e^{-\beta f(0) }\left[\frac{1}{f''(0)}\beta^{-1}-\frac{1 }{2}\sqrt{\frac{\pi}{2}}\frac{f'''(0)}{f''(0)^{5/2}}\beta^{-3/2}+o(\beta^{-3/2})\right].
\end{equation}
\item[(b)] 
Then,
\begin{equation}
\lim_{\beta\to\infty}\langle g \rangle_\beta^f=g(0).
\end{equation}
\item[(c)] If $g\in C^2 ([0,\pi/2];\mathbb{R})$ we have, as $\beta\to\infty$,
\begin{equation}\label{LaplaceE}
\langle g \rangle_\beta^f=g(0)+\sqrt{\frac{\pi}{2}} \frac{g'(0)}{f''(0)^{1/2}}\beta^{-1/2}
+\left[ \frac{g''(0)}{f''(0)}+\frac{3\pi-16}{12}\frac{g'(0)f'''(0)}{f''(0)^2} \right]\beta^{-1}+o(\beta^{-1}).
\end{equation}
\end{itemize}
\end{lem}

\begin{proof}
(a) We first write
\begin{equation*}
\int_{0}^{\pi/2}\e^{-\beta f(\theta) }\sin\theta\diff{\theta} = \e^{-\beta f(0) } \int_{0}^{\pi/2}\e^{-\beta (f(\theta) -f(0))}\sin\theta\diff{\theta}.
\end{equation*}
By assumption, there is $\epsilon>0$ such that, for all $0<\theta<\epsilon$, $f(\theta)- f(0) \ge \frac{1}{4}f''(0)\epsilon^2$. Hence,
\begin{equation*}
\int_\epsilon^{\pi/2}\e^{-\beta (f(\theta) - f(0))}\sin\theta\diff \theta \le \bigg(\sup_{\theta\in[\epsilon,\pi/2]}\e^{-\frac{\beta}{2} (f(\theta) - f(0))}\bigg) \int_\epsilon^{\pi/2}\e^{-\frac{\beta}{2} (f(\theta) - f(0))}\sin\theta\diff \theta \le C \e^{-\frac{\beta}{8} f''(0)\epsilon^2},
\end{equation*}
where $C$ is independent of $\beta,\epsilon$. Let now $\epsilon = \beta^{-1/4}$. Then $\beta \epsilon^2 = \beta^{1/2}\to\infty$ and the above integral vanishes exponentially as $\beta\to\infty$. In a neighbourhood of the minimum, Taylor expansions yield
\begin{align*}
\int_0^\epsilon\e^{-\beta (f(\theta) - f(0))}\sin\theta\diff \theta 
&= \int_0^\epsilon \e^{-\frac{\beta}{2}f''(0)\theta^2}\left(1 - \frac{\beta}{6}f'''(0) \theta^3 + \beta o(\theta^3)\right)
\left(\theta + o(\theta^2)\right)\diff \theta \\
&= \beta^{-1/2}\int_0^{\beta^{1/2}\epsilon} \e^{-\frac{1}{2}f''(0)\zeta^2}\left(\beta^{-1/2}\zeta - \beta^{-1} \frac{f'''(0)}{6} \zeta^4 + o(\beta^{-1})\right)\diff \zeta ,
\end{align*}
where we let $\theta = \beta^{-1/2}\zeta$. Since $\beta^{1/2}\epsilon \to\infty$, the error in replacing the upper bound of integration by $\infty$ is exponentially small indeed. The integrals can finally be carried out explicitly to yield~(\ref{laplace1}).

\noindent (b) Here again, we start by writing
\begin{equation}\label{g-g0}
\langle g\rangle_\beta^f = g(0) + \langle g - g(0)\rangle_\beta^f = g(0) + \frac{\int_0^{\pi/2}\e^{-\beta(f(\theta)-f(0))}(g(\theta) - g(0))\sin\theta\diff\theta}{\int_0^{\pi/2}\e^{-\beta(f(\theta)-f(0))}\sin\theta\diff\theta}.
\end{equation}
Proceeding as in~(a), we can restrict our attention to $[0,\beta^{-1/4})$ in the integrals. By continuity of $g$, for any $\tilde\epsilon>0$, there is $\beta<\infty$ such that $\theta<\beta^{-1/4}$ implies $\vert g(\theta) - g(0)\vert< \tilde\epsilon$. Hence,
\begin{equation*}
\big\vert \langle g\rangle_\beta^f - g(0) \big\vert = \big\vert \langle g - g(0)\rangle_\beta^f\big\vert \le \tilde\epsilon.
\end{equation*}

\noindent(c) We consider again~(\ref{g-g0}). Proceeding as in~(a), we expand the relevant part of the numerator as
\begin{align*}
\int_0^\epsilon &\e^{-\beta(f(\theta)-f(0))}(g(\theta) - g(0))\sin\theta\diff\theta
\\ 
&= \beta^{-1/2}\int_0^{\beta^{1/2}\epsilon} \e^{-\frac{1}{2}f''(0)\zeta^2}
\left(g'(0)\zeta^2 \beta^{-1} + \frac{1}{2}g''(0)\zeta^3\beta^{-3/2} - \frac{1}{6}f'''(0)g'(0)\beta^{-3/2}\zeta^5 + o(\beta^{-3/2})\right)\diff\zeta \\
&= \sqrt{\frac{\pi}{2}}\frac{g'(0)}{f''(0)^{3/2}}\beta^{-3/2} + \left[\frac{g''(0)}{f''(0)^{2}} - \frac{4}{3}\frac{g'(0)f'''(0)}{f''(0)^{3}}\right]\beta^{-2} + o(\beta^{-2}).
\end{align*}
Combining this with the expansion~(\ref{laplace1}), we obtain the claim.
\end{proof}

\end{document}